\documentclass{article}
\usepackage[utf8]{inputenc}

\usepackage[bookmarks,colorlinks,breaklinks]{hyperref}
\hypersetup{urlcolor=blue, colorlinks=true, citecolor=green!50!black, linkcolor=blue}
\usepackage[letterpaper, left=1in, right=1in, top=0.9in, bottom=0.9in]{geometry}
\usepackage[utf8]{inputenc}
\usepackage[american]{babel}

\usepackage[normalem]{ulem}

\usepackage{graphicx}
\graphicspath{{images/}{../images/}}

\usepackage{amsmath, amssymb, cases}
\usepackage{theoremstyles}
\usepackage{cleveref}
\usepackage{enumitem}
\usepackage{mdframed}
\usepackage{bbm}
\usepackage{bm}
\usepackage{tabu}

\usepackage[ruled]{algorithm}
\usepackage[noend]{algpseudocode}

\usepackage{microtype}

\usepackage{mdframed}
\usepackage{xcolor}

\usepackage{modletters}

\usepackage{mathtools}

\newcommand\restr[2]{{
  \left.\kern-\nulldelimiterspace 
  #1 
  \vphantom{\big|} 
  \right|_{#2} 
  }}

\makeatletter
\def\moverlay{\mathpalette\mov@rlay}
\def\mov@rlay#1#2{\leavevmode\vtop{%
   \baselineskip\z@skip \lineskiplimit-\maxdimen
   \ialign{\hfil$\m@th#1##$\hfil\cr#2\crcr}}}
\newcommand{\charfusion}[3][\mathord]{
    #1{\ifx#1\mathop\vphantom{#2}\fi
        \mathpalette\mov@rlay{#2\cr#3}
      }
    \ifx#1\mathop\expandafter\displaylimits\fi}
\makeatother

\usepackage{graphicx}

\newcommand{\size}[1]{\mathrm{size}}

\newcommand{\set}[2][ ]{\{#2 \ifthenelse{\equal{#1}{ }}{ }{~|~#1}\}}

\newcommand{\comment}[1]{}

\newcommand{\seepage}[2][See]{
    \marginnote{
        \scriptsize {#1} p.~\pageref{#2}
    }
}

\newcommand{\reuse}[1]{
	\expandafter\stepcounter{#1_help}
    \expandafter\label{#1_app}
    \csname#1\endcsname*
}

\newcommand{\trim}{\textsc{Trim}}
\newcommand{\shave}{\textsc{Shave}}
\newcommand{\vol}{\mathrm{vol}}
\algnewcommand\algorithmicwhen{\textbf{when}}
\algblockdefx[WHEN]{When}{EndWhen}[1]
  {\algorithmicwhen\ #1\  \algorithmicdo}
  {\algorithmicend\ \algorithmicwhen}

\makeatletter
\ifthenelse{\equal{\ALG@noend}{t}}%
  {\algtext*{EndWhen}}
  {}%
\makeatother

\usepackage{xcolor}

\newcommand{\ShowComment}{false}
\ifdefined\ShowComment
\newcommand{\danupon}[1]{{\color{green} DANUPON: #1}}
\newcommand{\knote}[1]{{\color{red}[{Karthik: \bf #1}]\marginpar{\color{red}*}}}
\newcommand{\sagnik}[1]{{\bf \color{blue} Sagnik: #1}}
\newcommand{\cnote}[1]{{\bf \color{purple} [{Calvin: \bf #1}] \marginpar{\color{purple}*}}}
\newcommand{\todos}[1]{{\color{red}[{TODO: \bf #1}]\marginpar{\color{red}*}}}
\else
\newcommand{\danupon}[1]{}
\newcommand{\knote}[1]{}
\newcommand{\sagnik}[1]{}
\newcommand{\cnote}[1]{}
\newcommand{\todos}[1]{}
\fi

\title{Faster connectivity in low-rank hypergraphs via expander decomposition}
\author{Calvin Beideman\thanks{University of Illinois, Urbana-Champaign, \texttt{\{calvinb2, karthe\}@illinois.edu}.} \and
Karthekeyan Chandrasekaran\footnotemark[1] \and 
Sagnik Mukhopadhyay\thanks{University of Copenhagen, Denmark, \texttt{sagnik@di.ku.dk}.} \and 
Danupon Nanongkai \thanks{University of Copenhagen and KTH Royal Institute of Technology, \texttt{danupon@gmail.com}.}
}
\date{}

\begin{document}

\begin{titlepage}
	\maketitle
	\pagenumbering{roman}

\begin{abstract}
The connectivity of a hypergraph is the minimum number of hyperedges whose deletion disconnects the hypergraph. 
We design an $\hat O_r(p + \min\{\lambda^{\frac{r-3}{r-1}} n^2, n^r/\lambda^{\frac{r}{r-1}}, \lambda^{\frac{5r-7}{4r-4}}n^{\frac{7}{4}}\})$\footnote{The $\hat O_r(\cdot)$ notation hides terms that are subpolynomial in the main parameter and terms that depend only on $r$.} time algorithm for computing hypergraph connectivity, 
where $p$ is the size, $n$ is the number of vertices, $r$ is the rank (size of the largest hyperedge), and $\lambda$ is the connectivity of the input hypergraph. Our algorithm also finds a minimum cut in the hypergraph. Our algorithm is faster than existing algorithms if $r = O(1)$ and $\lambda = n^{\Omega(1)}$. 
The heart of our algorithm is a structural result showing a trade-off between the number of hyperedges taking part in all minimum cuts and the size of the smaller side of any  minimum cut. 
This structural result can be viewed as a generalization of a well-known structural theorem for simple graphs [Kawarabayashi-Thorup, JACM 19 (Fulkerson Prize 2021)]. 
We extend the framework of expander decomposition to hypergraphs to prove this structural result. 
In addition to the expander decomposition framework, our faster algorithm also relies on a new near-linear time procedure to compute connectivity when one of the sides in a minimum cut is small. 
\end{abstract}

	\newpage
	\setcounter{tocdepth}{2}
	\newpage
\end{titlepage}

\newpage
\pagenumbering{arabic}

\section{Introduction}

A hypergraph $G=(V,E)$ is specified by a vertex set $V$ and a collection $E$ of hyperedges, where each hyperedge $e\in E$ is a subset of vertices. 
In this work, we address the problem of computing connectivity/global min-cut in hypergraphs with low rank  (e.g., constant rank). 
The \emph{rank} of a hypergraph, denoted $r$, is the size of the largest hyperedge---in particular, if the rank of a hypergraph is $2$, then the hypergraph is a graph. 
In the global min-cut problem, the input is a hypergraph with hyperedge weights $w:E\rightarrow \bbR_+$, and the goal is to find a minimum weight subset of hyperedges whose removal disconnects the hypergraph.  Equivalently, the goal is to find a partition of the vertex set $V$ into two non-empty parts $(C, V\setminus C)$ so as to minimize the weight of the set of hyperedges intersecting both parts. For a subset $C\subseteq V$, we will denote the weight of the set of hyperedges intersecting both $C$ and $V\setminus C$ by $d(C)$, the resulting function $d:V\rightarrow \bbR_+$ as the cut function of the hypergraph, and the weight of a min-cut by $\lambda(G)$ (we will use $\lambda$ when the graph $G$ is clear from context). 

If the input hypergraph is {\em simple}---i.e., each hyperedge has unit weight and no parallel copies---then 
the weight of a min-cut is also known as the {\em connectivity} of the hypergraph. 
We focus on finding connectivity in hypergraphs. 
We emphasize that, in contrast to graphs whose representation size is the number of edges, the representation size of a hypergraph $G=(V,E)$ is $p:=\sum_{e\in E}|e|$. We note that $p\le rm$, where $r$ is the rank and $m$ is the number of hyperedges in the hypergraph, and moreover, $r\le n$, where $n$ is the number of vertices. We emphasize that the number of hyperedges $m$ in a hypergraph could be exponential in the number of vertices. 

\paragraph{Previous Work.} Since the focus of our work is on simple unweighted hypergraphs, we discuss previous work for computing global min-cut in simple unweighted hypergraphs (i.e., computing connectivity) here (see Section \ref{sec:relevant-work} for a discussion of previous works on computing global min-cut in weighted hypergraphs/graphs). 
The current fastest algorithms to compute graph connectivity (i.e., when $r=2$) are randomized and run in time $\tO(m)$ \cite{Kar00, KawarabayashiT19, HenzingerRW17, GhaffariNT20, MukhopadhyayN20, GMW19}. In contrast, algorithms to compute hypergraph connectivity are much slower. Furthermore, for hypergraph connectivity/global min-cut, the known randomized approaches are not always faster than the known deterministic approaches. There are two broad algorithmic approaches for global min-cut in hypergraphs: vertex-ordering and random contraction. We discuss these approaches now. 

Nagamochi and Ibaraki \cite{NagamochiI92} introduced a groundbreaking vertex-ordering approach to solve global min-cut in graphs in time $O(mn)$. In independent works, Klimmek and Wagner \cite{KW96} as well as Mak and Wong \cite{MW00} gave two different generalizations of the vertex-ordering approach to compute hypergraph connectivity in $O(pn)$ time. Queyranne \cite{queyranne98} generalized the vertex-ordering approach further to solve
\emph{non-trivial symmetric submodular minimization}.\footnote{The input here is a symmetric submodular function $f:2^V\rightarrow \bbR$ via an evaluation oracle and the goal is to find a partition of $V$ into two non-empty parts $(C, V\setminus C)$ to minimize $f(C)$. We recall that a function $f:2^V\rightarrow \bbR$ is symmetric if $f(A) = f(V\setminus A)$ for all $A\subseteq V$ and is submodular if $f(A) + f(B) \ge f(A\cap B) + f(A\cup B)$ for all $A, B\subseteq V$. The cut function of a hypergraph $d:V\rightarrow \bbR_+$ is symmetric and submodular.}
Queyranne's algorithm can be implemented to compute hypergraph connectivity in $O(pn)$ time. Thus, all three vertex-ordering based approaches to compute hypergraph connectivity have a run-time of $O(pn)$. This run-time was improved to $O(p+\lambda n^2)$ by Chekuri and Xu \cite{ChekuriX18}: They designed an $O(p)$-time algorithm to construct a \emph{min-cut-sparsifier}, namely a subhypergraph $G'$ of the given hypergraph with size $p'=O(\lambda n)$ 
such that $\lambda(G')=\lambda(G)$. Applying the vertex-ordering based algorithm to $G'$ gives the connectivity of $G$ within a run-time of $O(p+\lambda n^2)$. 

We emphasize that all algorithms discussed in the preceding paragraph are deterministic. Karger \cite{Kar93} introduced the influential random contraction approach to solve global min-cut in graphs which was adapted by Karger and Stein \cite{KS96} to design an $\tilde{O}(n^2)$ time algorithm\footnote{For functions $f(n)$ and $g(n)$ of $n$, we say that $f(n) = \tilde{O}(g(n))$ if $f(n) = O(g(n)\text{polylog}(n))$ and $f(n) = \hat{O}(g(n))$ if $f(n) = O(g(n)^{1+o(1)})$, where the $o(1)$ is with respect to $n$. We say that $f(n) = O_r(g(n))$ if $f(n) = O(g(n)h(r))$ for some function $h$. We define $\tilde{O}_r(f(n))$ and $\hat{O}_r(f(n))$ analogously.
}. 
Kogan and Krauthgamer \cite{KK15} extended the random contraction approach to solve global min-cut in $r$-rank hypergraphs in time $\tilde{O}_r(mn^2)$. 
Ghaffari, Karger, and Panigrahi \cite{GKP17} suggested a non-uniform distribution for random contraction in hypergraphs and used it to design an algorithm to compute hypergraph connectivity 
in $\tilde{O}((m+\lambda n)n^2)$ time. 
Chandrasekaran, Xu, and Yu \cite{CXY19} refined their non-uniform distribution to obtain an $O(pn^3\log{n})$ time algorithm for global min-cut in hypergraphs. 
Fox, Panigrahi, and Zhang \cite{FPZ19} 
proposed a branching approach to exploit the refined distribution leading to an $O(p+n^r\log^2{n})$ time algorithm for hypergraph global min-cut, where $r$ is the rank of the input hypergraph. Chekuri and Quanrud \cite{CQ21} designed an algorithm based on isolating cuts which achieves a runtime of $\tilde{O}(\sqrt{pn(m+n)^{1.5}})$ for global min-cut in hypergraphs.

Thus, the current fastest known algorithm to compute hypergraph connectivity is a combination of the algorithms of Chekuri and Xu \cite{ChekuriX18}, Fox, Panigrahi, and Zhang \cite{FPZ19}, and Chekuri and Quanrud \cite{CQ21} with a run-time of 
$$\tilde{O}\left(p+\min\left\{\lambda n^2, n^r, \sqrt{pn(m+n)^{1.5}}\right\}\right).$$ 

\subsection{Our results}
In this work, we improve the run-time to compute hypergraph connectivity in low rank simple hypergraphs. 

\begin{restatable}{theorem}{minCutAlgorithm}\label{thm:min-cut-algorithm}[Algorithm]\label{theorem:algo}
Let $G$ be  an $r$-rank $n$-vertex simple hypergraph of size $p$. Then, there exists a randomized algorithm that takes $G$ as input and runs in time 
\[\hat{O}_r \left(p + \min\left\{ \lambda^{\frac{r-3}{r-1}} n^2, \frac{n^{r}}{\lambda^{\frac{r}{r-1}}}, \lambda^{\frac{5r-7}{4r-4}}n^{\frac{7}{4}} \right\} \right)\]
to return the connectivity $\lambda$ of $G$ with high probability. Moreover, the algorithm returns a min-cut in $G$ with high probability. 
\end{restatable}

See Section \ref{sec:main-algorithm} for the run-time of our algorithm in the $O$-notation.
Our techniques can also be used to obtain a deterministic algorithm that runs in time 
\[
\hat{O}_r\left(p+\min\left\{\lambda n^2, \lambda^{\frac{r-3}{r-1}}n^2+\frac{n^r}{\lambda}\right\}\right).  
\]
Our deterministic algorithm is faster than Chekuri and Xu's algorithm when $r$ is a constant and $\lambda = \Omega(n^{(r-2)/2})$, while our randomized algorithm is faster than known algorithms if $r$ is a constant and $\lambda = n^{\Omega(1)}$.
We summarize the previous fastest algorithms and our results in Table \ref{tab:results}. 

\tabulinesep=1.2mm
\begin{table}[ht]
\centering
\begin{tabu}{|c|c|c|}
\hline
 & \textbf{Deterministic} & \textbf{Randomized}\\
 \hline
 Previous run-time & $O(p+\lambda n^2)$ \cite{ChekuriX18} & $\tO(p+\min\{\lambda n^2, n^r, \sqrt{pn(m+n)^{1.5}}\}\})$\\ &&\cite{ChekuriX18, FPZ19, CQ21}\\
 \hline
 Our run-time & $\hat{O}_r\left(p+\min\left\{\lambda n^2,  \lambda^{\frac{r-3}{r-1}}n^2+\frac{n^r}{\lambda}\right\}\right)$ &$\hat{O}_r \left(p + \min\left\{ \lambda^{\frac{r-3}{r-1}} n^2, \frac{n^{r}}{\lambda^{\frac{r}{r-1}}}, \lambda^{\frac{5r-7}{4r-4}}n^{\frac{7}{4}} \right\} \right)$\\
 \hline
\end{tabu}
\caption{Comparison of results to compute hypergraph connectivity (simple unweighted $r$-rank $n$-vertex $m$-hyperedge $p$-size hypergraphs with connectivity $\lambda$).} \label{tab:results}
\end{table}

Our algorithm for Theorem \ref{theorem:algo} proceeds by considering two cases: either (i) the hypergraph has a min-cut where one of the sides is small or 
 (ii) both sides of every min-cut in the hypergraph are large.
To account for case (i), we design a near-linear time algorithm to compute a min-cut; to account for case (ii), we perform contractions to reduce the size of the hypergraph without destroying a min-cut and then run known algorithms on the smaller-sized hypergraph leading to savings in run-time. Our contributions in this work are twofold: (1) On the algorithmic front, we design a near-linear time algorithm to find a min-cut where one of the sides is small (if it exists); (2) On the structural front, we show a trade-off between the number of hyperedges taking part in all minimum cuts and the size of the smaller side of any minimum cut (see Theorem \ref{thm:number-of-hyperedges-in-mincuts}). 
This structural result is a generalization of the acclaimed Kawarabayashi-Thorup graph structural theorem \cite{KT15, KawarabayashiT19} (Fulkerson prize 2021). 
We use the structural result to reduce the size of the hypergraph in case (ii). We elaborate on this structural result now. 



\begin{restatable}{theorem}{numEdgesMincuts}
\label{thm:number-of-hyperedges-in-mincuts}[Structure]
Let $G=(V,E)$ be an $r$-rank $n$-vertex simple hypergraph with $m$ hyperedges and connectivity $\lambda$. Suppose 
$\lambda \ge r(4r^2)^r$. Then, at least one of the following holds:\begin{enumerate}
    \item \label{itm:small-cut} There exists a min-cut $(C, V \setminus C)$ such that \[\min\{|C|, |V \setminus C|\}\le r-\frac{\log{(\frac{\lambda}{4r})}}{\log{n}},\]
    
    \item \label{itm:few-edges} The number of hyperedges in the union of all min-cuts is 
    \[O\left(r^{9r^2+2}\left(\frac{6r^2 }{\lambda}\right)^{\frac{1}{r-1}}m\log n\right)
    = \tilde{O}_r\left(\frac{m}{\lambda^{\frac{1}{r-1}}}\right).\]
\end{enumerate}
\end{restatable}

The Kawarabayashi-Thorup structural theorem for graphs \cite{KT15, KawarabayashiT19} states that if every min-cut is non-trivial, then the number of edges in the union of all min-cuts is $O(m/\lambda)$, 
where a cut is defined to be \emph{non-trivial} if it has at least two vertices on each side. Substituting $r=2$ in our structural theorem recovers this known Kawarabayashi-Thorup structural theorem for graphs. 
We emphasize that the Kawarabayashi-Thorup structural theorem for graphs is the backbone of the current fastest algorithms for computing connectivity in graphs and has been proved in the literature via several different techniques \cite{KawarabayashiT19, HenzingerRW17, RubinsteinSW18, GhaffariNT20, Sar20}. 
Part of the motivation behind our work was to understand whether the Kawarabayashi-Thorup structural theorem for graphs could hold for constant rank hypergraphs and if not, then what would be an appropriate generalization. 
We discovered that the Kawarabayashi-Thorup graph structural theorem \emph{does not} hold for hypergraphs: There exist hypergraphs in which (i) the min-cut capacity $\lambda$ is $\Omega(n)$, (ii) there are no trivial min-cuts, and (iii) the number of hyperedges in the union of all min-cuts is a constant fraction of the number of hyperedges---see Lemma \ref{lem:hypergraph-nontrivial} in Appendix \ref{appn:hypergraph-nontrivial} for such an example. The existence of such examples suggests that we need an alternative definition of \textit{trivial min-cuts} if we hope to extend the Kawarabayashi-Thorup structural theorem for graphs to $r$-rank hypergraphs. Conclusion \ref{itm:small-cut} of Theorem \ref{thm:number-of-hyperedges-in-mincuts} can be viewed as a way to redefine the notion of \emph{trivial min-cuts}. We denote the \emph{size} of a cut $(C, V\setminus C)$  to be $\min \{|C|, |V\setminus C|\}$---we emphasize that the size of a cut refers to the size of the smaller side of the cut as opposed to the capacity of the cut. A min-cut is \emph{small-sized} if the smaller side of the cut has at most $r - \log (\lambda/4r)/\log n$ many vertices. With this definition, Conclusion \ref{itm:few-edges} of Theorem \ref{thm:number-of-hyperedges-in-mincuts} can be viewed as a generalization of the Kawarabayashi-Thorup structural theorem to hypergraphs which have no small-sized min-cuts: it says that if no min-cut is small-sized, then the number of hyperedges in the union of all min-cuts is $\tilde{O}_r(m/\lambda^{\frac{1}{r-1}})$. 

We mention that the factor $\lambda^{-1/(r-1)}$ in Conclusion \ref{itm:few-edges} of Theorem \ref{thm:number-of-hyperedges-in-mincuts} cannot be improved: There exist hypergraphs in which every min-cut has at least $\sqrt{n}$ vertices on both sides and the number of hyperedges in the union of all min-cuts is $\Theta(m \cdot \lambda^{-1/(r-1)})$---see \Cref{lem:tight-condn-2} in Appendix \ref{appn:tight-condn-2}. We also note that the structural theorem holds only for \emph{simple} hypergraphs/graphs and is known to fail for weighted graphs. As a consequence, our algorithmic techniques are applicable only in simple hypergraphs and not in weighted hypergraphs. 

\subsection{Technical overview} \label{sec:tech-overview}

Concepts used in the proof strategy of Theorem \ref{thm:number-of-hyperedges-in-mincuts} will be used in the algorithm of Theorem \ref{thm:min-cut-algorithm} as well, so it will be helpful to discuss the proof strategy of Theorem \ref{thm:number-of-hyperedges-in-mincuts} before the algorithm. We discuss this now. 
We define 
a cut $(C, V\setminus C)$ to be \emph{moderate-sized} if 
$\min\{|C|, |V\setminus C|\}\in(r-\log(\lambda/4r)/\log n, 4r^2)$
and to be \emph{large-sized} if $\min\{|C|, |V\setminus C|\}\ge 4r^2$; we recall that the cut $(C, V\setminus C)$  is small-sized if $\min\{|C|, |V\setminus C|\}\le r-\log{(\lambda/4r)}/\log{n}$. 

\paragraph{Proof strategy for the structural theorem (Theorem \ref{thm:number-of-hyperedges-in-mincuts}).}
We assume that $\lambda>r(4r^2)^r$ as in the statement of Theorem \ref{thm:number-of-hyperedges-in-mincuts}. 
The first step of our proof is to show that every min-cut in a hypergraph is either large-sized or small-sized but not moderate-sized---in particular, we prove that if $(C, V\setminus C)$ is a min-cut with $\min\{|C|, |V\setminus C|\} < 4r^2$, then it is in fact a small-sized min-cut (Lemma \ref{cor:small-side-not-tiny-implies-small-side-large} with the additional assumption that $\lambda>r(4r^2)^r$). Here is the informal argument: For simplicity, we will show that if $(C, V\setminus C)$ is a min-cut with $\min\{|C|, |V\setminus C|\} < 4r^2$, then $\min\{|C|, |V\setminus C|\} \leq r$. For the sake of contradiction, suppose that $\min\{|C|, |V\setminus C|\} > r$.  The crucial observation is that since the hypergraph has rank $r$, no hyperedge can contain the smaller side of the min-cut entirely. The absence of such hyperedges means that even if we pack hyperedges in $G$ as densely as possible while keeping $(C,V\setminus C)$ as a min-cut, we cannot pack sufficiently large number of hyperedges to ensure that the degree of each vertex is at least $\lambda$. A more careful counting argument extends this proof approach to show that $\min\{|C|, |V\setminus C|\} \leq r -\log \lambda/\log n$.

Now, in order to prove Theorem \ref{thm:number-of-hyperedges-in-mincuts}, it suffices to prove Conclusion \ref{itm:few-edges} under the assumption that all min-cuts are large-sized, i.e., $\min\{|C|, |V\setminus C|\}\ge 4r^2$ for every min-cut $(C, V\setminus C)$. Our strategy to prove Conclusion \ref{itm:few-edges} is to find a partition of the vertex set $V$ such that (i) every hyperedge that is completely contained in one of the parts does not cross any min-cut, and (ii) the number of hyperedges that intersect multiple parts (and therefore, possibly cross some min-cut) is small, i.e., $\tO_r(m \cdot\lambda^{-1/(r-1)})$. To this end, we start by partitioning the vertex set of the hypergraph $G$ into $X_1, \ldots, X_k$ 
such that the total number of hyperedges intersecting more than one part of the partition is $\tO_r(m \cdot\lambda^{-1/(r-1)})$ and the subhypergraph induced by each $X_i$ has conductance $\Omega_r(\lambda^{-1/(r-1)})$ (see Section \ref{sec:prelims} for the definition of conductance)---such
 a decomposition is known as an \emph{expander decomposition}. An expander decomposition immediately satisfies (ii) since the number of hyperedges intersecting more than one part is small. Unfortunately, it may not satisfy (i); yet,  it is very close to satisfying (i)---we can guarantee that 
 for every min-cut $(C, V \setminus C)$ and every $X_i$, either $C$ includes very few vertices from $X_i$, or $C$ includes almost all the vertices of $X_i$
 i.e., $\min\{|X_i \cap C|, |X_i \setminus C|\} = O_r(\lambda^{1/(r-1)})$. 
 We note that if $\min\{|X_i \cap C|, |X_i \setminus C|\} = 0$ for every min-cut $(C, V \setminus C)$ and every part $X_i$ then (i) would be satisfied; moreover, if a part $X_j$ is a singleton vertex part (i.e., $|X_j|=1$), then $\min\{|X_j \cap C|, |X_j \setminus C|\} = 0$ holds. So, our strategy, at this point, is to remove some of the vertices from $X_i$ to form their own singleton vertex parts in the partition in order to achieve $\min\{|X_i \cap C|, |X_i \setminus C|\} = 0$ while controlling the increase in the number of hyperedges that cross the parts. This is achieved by a \trim\ operation and a series of \shave\ operations. 

The crucial parameter underlying \trim\ and \shave\ operations is the notion of degree within a subset:  We will denote the degree of a vertex $v$ as $d(v)$ and define the degree contribution of a vertex $v$ inside a vertex set $X$, denoted by $d_X(v)$, to be the number of hyperedges containing $v$ that are completely contained in $X$. 
The \trim\ operation on a part $X_i$ repeatedly removes from $X_i$ vertices with small degree contribution inside $X_i$, i.e., $d_{X_i}(v) < d(v)/2r$ until no such vertex can be found. Let $X'_i$ denote the set obtained from $X_i$ after the \trim\ operation.
We note that our partition now consists of $X_1', \ldots, X_k'$ 
as well as singleton vertex parts for each vertex that we removed with the $\trim$ operation. 
This operation alone makes a lot of progress towards our goal---we show that $\min\{|X'_i \cap C|, |X'_i \setminus C|\} = O(r^2)$, while the number of hyperedges crossing the partition blows up only by an $O(r)$ factor (see Claims \ref{clm:trim-bound} and \ref{clm:hyperedges-lost-to-trim}). The little progress that is left to our final goal is achieved by a series of ($O(r^2)$ many) \shave\ operations. The \shave\ operation finds the set of 
vertices in each $X'_i$ whose degree contribution inside $X'_i$ is not very large, i.e., $d_{X'_i}(v) \leq (1-r^{-2})d(v)$ and removes this set of vertices from $X'_i$ in one shot---such vertices are again declared as singleton vertex parts in the partition. We show that the \shave\ operation strictly reduces $\min\{|X'_i \cap C|, |X'_i \setminus C|\}$ without adding too many hyperedges across the parts (see Claims \ref{clm:shave-bound} and \ref{clm:hyperedges-lost-to-shave})---this argument crucially uses the assumption that all min-cuts are large-sized (i.e., $\min\{|C|,|V \setminus C|\} \geq 4r^2$). Because of our guarantee from the \trim\ operation regarding $\min\{|X'_i \cap C|, |X'_i \setminus C|\}$, we need to perform the \shave\ operation $O(r^2)$ times to obtain a partition that satisfies conditions (i) and (ii) stated in the preceding paragraph.

\paragraph{Algorithm from structural theorem (Theorem \ref{theorem:algo}).} 
We now describe how to use the structural theorem (Theorem \ref{thm:number-of-hyperedges-in-mincuts}) to design an algorithm for hypergraph connectivity with run-time guarantee as stated in Theorem \ref{theorem:algo}. 
If the min-cut capacity $\lambda$ is small, then existing algorithms are already fast, so we assume that the min-cut capacity $\lambda$ is sufficiently large. 
With this assumption, we consider the two conclusions mentioned in Theorem \ref{thm:number-of-hyperedges-in-mincuts}. 
We design a near-linear time algorithm for finding small-sized min-cuts if the first conclusion holds. If the first conclusion fails, then the second conclusion holds; now, in order to obtain an algorithm for min-cut, we make our proof of Theorem \ref{thm:number-of-hyperedges-in-mincuts} constructive---we obtain a small-sized hypergraph which retains all hyperedges that participate in min-cuts; we run existing min-cut algorithms on this small-sized hypergraph. 
We note here that, in a recent result, Saranurak \cite{Sar20} used a constructive version of a similar theorem to design an algorithm for computing connectivity in simple graphs in almost linear time (with a run-time of $m^{1+o(1)}$). Since our structural theorem is meant for hypergraph connectivity (and is hence, more complicated than what is used by \cite{Sar20}), we have to work more. 

We now briefly describe our algorithm: 
Given an $r$-rank hypergraph $G$, we estimate the connectivity $\lambda$ to within a constant factor in $O(p)$ time using an algorithm of Chekuri and Xu \cite{ChekuriX18}. Next, we run the sparsification algorithm of \cite{ChekuriX18} in time $O(p)$ to obtain a new hypergraph $G'$ with size $p'=O(\lambda n)$ such that all min-cuts are preserved. 
The rest of the steps are run on this new hypergraph $G'$. 
We have two possibilities as stated in Theorem \ref{thm:number-of-hyperedges-in-mincuts}. We account for these two  possibilities by running two different algorithms: (i) Assuming that some min-cut has size less than $r-\log (\lambda/4r)/\log n$, we design a near-linear time algorithm to find a min-cut (see Theorem \ref{thm:new-small-mincut-algo}). This algorithm is inspired by recent vertex connectivity algorithms, in particular the local vertex connectivity algorithm of \cite{FNYSY20,NanongkaiSY19} and the sublinear-time kernelization technique of \cite{LNPSY21}. This algorithm runs in $\tilde{O}_r(p)$ time. (ii) Assuming that every min-cut is large-sized, we design a fast algorithm to find a min-cut (see Theorem \ref{thm:large-mincut-algo}). For this, we find an expander decomposition $\mathcal{X}$ of $G'$, perform a \trim\ operation followed by a series of $O(r^2)$ \shave\ operations, and then contract each part of the trimmed and shaved expander decomposition to obtain a hypergraph $G''$. This reduces the number of vertices in $G''$ to $O_r(n/\lambda^{1/(r-1)})$ and consequently, running the global min-cut algorithm of either \cite{FPZ19} or \cite{CX17} or \cite{CQ21} (whichever is faster) on $G''$ leads to an overall run-time of $\hat{O}_r(p + \min\{ \lambda^{(r-3)/(r-1)}n^2, n^r/\lambda^{r/(r-1)}, \lambda^{(5r-7)/(4r-4)}n^{7/4} \})$ for step (ii). We return the cheaper of the two cuts found in steps (i) and (ii). The correctness of the algorithm follows by the structural theorem and the total run-time is $\hat{O}_r(p + \min\{ n^r/\lambda^{r/(r-1)}, \lambda^{(r-3)/(r-1)}n^2, \lambda^{(5r-7)/(4r-4)}n^{7/4} \})$.

We note here that the expander decomposition framework for graphs was developed in a series of works for the dynamic connectivity problem \cite{ns17, w17, sw19, cglnps20}. Very recently, it has found applications for other problems \cite{GoranciRST20hierarchy,BernsteinBNPSS20, BernsteinGS20scc}. Closer to our application, Saranurak \cite{Sar20} used expander decomposition to give an algorithm to compute edge connectivity in graphs via the use of \trim\ and \shave\ operations. The \trim\ and \shave\ operations were introduced by Kawarabayashi and Thorup \cite{KawarabayashiT19} to compute graph connectivity in deterministic $O(m\log ^{12} n)$ time. Our line of attack is an adaptation of Saranurak's approach. 

\paragraph{Organization.} In Section \ref{sec:relevant-work}, we discuss algorithms for the more general problem of weighted hypergraph min-cut. In Section \ref{sec:prelims}, we introduce the necessary preliminaries. In Section \ref{sec:structural-theorem-main}, we prove our structural theorem (Theorem \ref{thm:number-of-hyperedges-in-mincuts}). 
We present our hypergraph min-cut algorithm and prove Theorem \ref{thm:min-cut-algorithm} in Section \ref{sec:min-cut-algo}. 
In Appendix \ref{sec:det-small-size-min-cut}, we give a deterministic variant of our min-cut algorithm from Section \ref{sec:small-size-min-cut} for small-sized min-cut---this is necessary for the (the slightly slower) deterministic version of our min-cut algorithm. We provide a tight example for our structural theorem (i.e., Theorem \ref{thm:number-of-hyperedges-in-mincuts}) in Appendix \ref{appn:tight-condn-2}. Lastly, in Appendix \ref{appn:hypergraph-nontrivial}, we provide an example to justify the modified notion of non-triviality while considering min-cuts in hypergraphs. For the technical sections, we encourage the reader to consider $r=O(1)$ during first read.

\subsection{Relevant Work}\label{sec:relevant-work}
We briefly describe the approaches known for global min-cut in weighted graphs and hypergraphs. We begin by discussing the case of graphs. 
On the deterministic front, until recently, the fastest known algorithm for global min-cut was $O(mn)$ due to Nagamochi and Ibaraki \cite{NagamochiI92}. This was recently improved to $\tilde{O}(m\cdot \min\{\sqrt{m}, n^{2/3}\})$ by Li and Panigrahi \cite{LP20}. We now discuss the randomized approaches. 
As mentioned earlier, Karger introduced the influential random contraction approach for graphs leading to a randomized algorithm that runs in time $\tilde{O}(mn^2)$. Karger and Stein refined this approach to obtain a run-time of $\tilde{O}(n^2)$. 
Subsequently, Karger gave a tree-packing approach that runs in time $O(m\log^3 n)$. The tree-packing approach has garnered much attention recently leading to run-time improvements \cite{MukhopadhyayN20, GMW19} and extensions to solve more general partitioning problems \cite{Th08, CQX19}. 

We now discuss the case of hypergraphs. We first focus on bounded rank hypergraphs. Fukunaga \cite{F10} generalized the tree-packing approach to $r$-rank hypergraphs to design an algorithm that runs in time $\tilde{O}_r(m^2 n^{2r-1})$. As mentioned earlier, Kogan and Krauthgamer \cite{KK15} generalized the random contraction approach to $r$-rank hypergraphs to design an algorithm that runs in time $O_r(mn^2)$. 
Fox, Panigrahi, and Zhang \cite{FPZ19} gave a randomized algorithm that runs in time $O(p+n^r\log^2 n)$. 
We now turn to arbitrary rank hypergraphs. 
As mentioned earlier, the approach of Nagamochi and Ibaraki was generalized in three different directions for arbitrary rank hypergraphs \cite{KW96, MW00, queyranne98} with all of them leading to a deterministic run-time of $O(pn)$. Chekuri and Xu improved this deterministic run-time to $O(p+\lambda n^2)$ for unweighted hypergraphs. In a series of works \cite{GKP17, CXY19, FPZ19}, the random contraction approach was refined and extended leading to a randomized run-time of $O(mn^2\log^3n)$. Interestingly, randomized approaches are slower than deterministic approaches for arbitrary rank hypergraphs. 

Motivated towards improving the $\log^3 n$ factor in the deterministic run-time for global min-cut in graphs, as well as obtaining a faster (randomized) algorithm, 
Kawarabayashi and Thorup \cite{KT15, KawarabayashiT19} initiated the study of global min-cut in simple unweighted graphs (i.e., computing connectivity). They gave a deterministic algorithm for computing connectivity that runs in time $O(m\log^{12} n)$. Henzinger, Rao, and Wang \cite{HenzingerRW17} improved the deterministic run-time to $O(m\log^2n \log\log^2n)$. Ghaffari, Nowicki, and Thorup \cite{GhaffariNT20} improved the randomized run-time to $O(m\log n)$ and $O(m+n\log^3 n)$. Saranurak \cite{Sar20} illustrated the use of expander decomposition to compute graph connectivity in time $m^{1+o(1)}$ time. All these algorithms are based on novel structural insights on graph cuts. 
As mentioned earlier, our work was motivated by the question of how to generalize the structural insights for graphs to constant rank hypergraphs. 

\section{Preliminaries}\label{sec:prelims}

Let $G=(V,E)$ be a hypergraph. Let $S, T\subseteq V$ be subsets of vertices. We define 
\begin{align*}
&E[S] := \{e \in E \colon e \subseteq S\},\\
&E(S,T) := \{ e \in E \colon e \subseteq S \cup T \text{ and } e \cap S, e \cap T \neq \emptyset \},  \text{ and}\\
&E^o(S, T) := \{ e \in E \colon e \cap S, e \cap T \neq \emptyset \}.
\end{align*}
We note that $E[S]$ is the set of hyperedges completely contained in $S$, $E(S,T)$ is the set of hyperedges contained in $S \cup T$ and intersecting both $S$ and $T$, and $E^o(S,T)$ is the set of hyperedges intersecting both $S$ and $T$. With this notation, if $S$ and $T$ are disjoint, then $E(S,T) = E[S \cup T] - E[S] - E[T]$ and moreover, if the hypergraph is a graph, then $E(S,T)=E^o(S,T)$. A cut is a partition $(S, V \setminus S)$ where both $S$ and $V \setminus S$ are non-empty. Let $\delta(S) := E(S, V \setminus S)$. For a vertex $v \in V$, we let $\delta(v)$ represent $\delta(\{v\})$. We define the capacity of $(S, V \setminus S)$ as $|\delta(S)|$, and call a cut as a min-cut if it has minimum capacity among all cuts in $G$. The connectivity of $G$ is the capacity of a min-cut in $G$.

We recall that the \emph{size} of a cut $(S, V \setminus S)$ is $\min\{|S|, |V \setminus S| \}$. We emphasize the distinction between the size of a cut and the capacity of a cut: size is the cardinality of the smaller side of the cut while capacity is the number of hyperedges crossing the cut. 

For a vertex $v\in V$ and a subset $S\subseteq V$, we define the degree of $v$ by $d(v):=|\delta(v)|$ and its degree inside $S$ by $d_S(v):=|e\in \delta(v):e\subseteq S|$. We define $\delta := \min_{v \in V} d(v)$ to be the minimum degree in $G$. We define $\vol(S) := \sum_{v \in S} d(v)$ and $\vol_S(T) := \sum_{v \in T} d_S(v)$. 
We define the \textit{conductance} of a set $X \subseteq V$ as $\min_{\emptyset \neq S \subsetneq X}\{\frac{|E^o(S, X \setminus S)|}{\min\{\vol(S),\vol(X \setminus S)\}}\}$.
For positive integers, $i < j$, we let $[i,j]$ represent the set $\{i, i+1, \ldots, j-1, j\}$.
The following proposition will be useful while counting hyperedges within nested sets. 
\begin{proposition}\label{prop:crossing-edges-volume-relationship}
Let $G=(V, E)$ be an $r$-rank $n$-vertex hypergraph and let $T\subseteq S\subseteq V$. Then, 
\[
\left|E(T, S \setminus T)\right| \geq \left(\frac{1}{r-1}\right)\left(\vol_{S}(T) - r\left|E[T]\right|\right).
\]
\end{proposition}

\begin{proof}
For a vertex $v \in V$ and a hyperedge $e \in E$, let $Y_{v, e} := 1$ if $v \in e$ and $e \subseteq S$, and $0$ otherwise. We note that
\begin{align*}
    \vol_S(T) &= \sum_{v \in T} d_{S}(v) 
    = \sum_{v \in T} \sum_{e \in E} Y_{v,e} 
    = \sum_{e \in E} \sum_{v \in T} Y_{v,e} 
    = \sum_{e \in E(T, S \setminus T)} \sum_{v \in T} Y_{v,e} + \sum_{e \subseteq T} \sum_{v \in T} Y_{v,e} \\
    &\leq \sum_{e \in E(T, S \setminus T)} (r-1) + \sum_{e \subseteq T} r 
    = (r-1)\left|E(T, S \setminus T)\right| + r\left|E[T]\right|.
\end{align*}
Rearranging this inequality gives the statement of the claim.
\end{proof}

For the structural theorem, we need the following folklore result showing the existence of an expander decomposition in graphs (e.g., see \cite{sw19}).

\begin{theorem} [Existential graph expander decomposition]\label{thm:expander-decomp-exists}
Let $G = (V,E)$ be an $n$-vertex graph (possibly multigraph) with $m$ edges and let $\phi \leq 1$ be a positive real value. Then, there exists a partition $\{X_1, \dots, X_k\}$ of the vertex set $V$ such that the following hold: 
\begin{enumerate}
    \item $\sum_{i=1}^k |\delta(X_i)| = O(\phi m\log{n})$ and
    \item For every $i \in [k]$ and every non-empty set $S \subseteq X_i$, we have that \[|E(S, X_i \setminus S)| \geq \phi \cdot \min\{\vol(S), \vol(X_i \setminus S)\}.\]
\end{enumerate}
\end{theorem}

Our algorithm for min-cut uses several other algorithms from the literature as subroutines. We state these algorithms and their guarantees here. To begin with, we will need a constructive version of Theorem \ref{thm:expander-decomp-exists} that we state below. 

\begin{theorem} [Algorithmic graph expander decomposition  \cite{cglnps20}]\label{thm:CGLNPS-expander-decomp} Let $G = (V,E)$ be a graph (possibly multigraph) with $m$ edges, and let $\phi \leq 1$ be a positive real value. Then, 
there exists a deterministic algorithm \textsc{GraphExpanderDecomp} which takes $G$ and $\phi$ as input and runs in time 
$O(m^{1+o(1)})$
to return a partition $\{X_1, \dots, X_k\}$ of the vertex set $V$ such that  the following hold: 
\begin{enumerate}
    \item $\sum_{i=1}^k |\delta(X_i)| = O(\phi m^{1+o(1)})$ and
    \item For every $i \in [k]$ and every nonempty set $S \subseteq X_i$, we have that $|E(S, X_i \setminus S)| \geq \phi \cdot \min\{\vol(S), \vol(X_i \setminus S)\}$.
\end{enumerate}
\end{theorem}
Throughout this work, $o(1)$ is with respect to the number of vertices $n$. 
The following result summarizes the vertex-ordering based algorithms for computing connectivity in hypergraphs. 
\begin{theorem}\cite{KW96, queyranne98, MW00}\label{thm:slow-min-cut} 
Let $G$ be an $n$-vertex
hypergraph of size $p$ (possibly multihypergraph). Then, there exists a deterministic algorithm \textsc{SlowMinCut} that takes $G$ as input and runs in time $O(pn)$ to compute the min-cut capacity $\lambda$ of $G$ (and also return a min-cut).
\end{theorem}

We need a fast estimator for connectivity in hypergraphs: Matula \cite{Matula93} designed a linear-time $(2+\epsilon)$-approximation for min-cut in graphs. Chekuri and Xu \cite{ChekuriX18} generalized Matula's approach to hypergraphs. 

\begin{theorem}\cite{ChekuriX18}\label{thm:CX-approximation} 
Let $G$ be a hypergraph of size $p$ (possibly multihypergraph). Then, there exists an algorithm $\textsc{CXApproximation}$ that takes $G$ as input and runs in time $O(p)$ to return a positive integer $k$ such that $\lambda < k \leq 3\lambda$, where $\lambda$ is the capacity of a min-cut in $G$. 
\end{theorem}

For most graph optimization algorithms, it is helpful to reduce the size of the input graph without losing an optimum before running slow algorithms. A $k$-certificate reduces the size of the hypergraph while ensuring that every cut capacity $|\delta(C)|$ does not fall below $\min \{k, |\delta(C)|\}$. 
The following result on $k$-certificates was shown by Guha, McGregor, and Tench \cite{GMT15}. The size of the $k$-certificate was subsequently improved by Chekuri and Xu via trimming hyperedges, but the Chekuri-Xu $k$-certificate could be a multihypergraph. Since we need a simple hypergraph as a $k$-certificate, we rely on the weaker result of \cite{GMT15}. 

\begin{theorem}\label{thm:k-certificate} \cite{GMT15, ChekuriX18}
Let $G=(V,E)$ be an $n$-vertex hypergraph (possibly multihypergraph) of size $p$, and let $k$ be a positive integer. Then, there exists an algorithm $\textsc{Certificate}$ that takes $G$ and $k$ as input and runs in time $O(p)$ to return a subhypergraph $G'=(V, E')$ with $\sum_{e \in E'} |e| = O(rnk)$ such that for every cut $(C,V \setminus C)$ in $G$, we have $|\delta_{G'}(C)| \geq \min\{k, |\delta_G(C)|\}$.
\end{theorem}

Chekuri and Xu \cite{ChekuriX18} noted that the slow min-cut algorithm of Theorem \ref{thm:slow-min-cut} can be sped up by using the $k$-certificate of Theorem \ref{thm:k-certificate}. We illustrate their approach in Algorithm \ref{alg:CX-min-cut}. Our algorithm can be viewed as an improvement of Algorithm \ref{alg:CX-min-cut}---we replace the call to $\textsc{SlowMinCut}$ with a call to a faster min-cut algorithm. 

\begin{algorithm}[ht]
\caption{{CXMinCut}($G$)}\label{alg:CX-min-cut}
\begin{algorithmic}
\State $k \gets \textsc{CXApproximation}(G)$
\State $G' \gets \textsc{Certificate}(G, k)$
\State \Return $\textsc{SlowMinCut}(G')$
\end{algorithmic}
\end{algorithm}

We briefly discuss the correctness and run-time of Algorithm \ref{alg:CX-min-cut}. By Theorem \ref{thm:CX-approximation}, we have that $k > \lambda$, so Theorem \ref{thm:k-certificate} guarantees that the min-cuts of $G'$ will be the same as the min-cuts of $G$. Thus, \textsc{SlowMinCut}$(G')$ will find a min-cut for $G$ and the algorithm will return it. We now bound the run-time: By Theorem \ref{thm:CX-approximation},  \textsc{CXApproximation} runs in $O(p)$ time to return an estimate $k$ that is at most $3\lambda$. By Theorem \ref{thm:k-certificate}, \textsc{Certificate} runs in $O(p)$ time to return a subhypergraph $G'=(V,E')$ with size $p'=\sum_{e\in E'} |e| = O(rkn) = O(r\lambda n)$. Consequently, $\textsc{SlowMinCut}$ runs in time $O(r\lambda n^2)$. Thus, the run-time of Algorithm \ref{alg:CX-min-cut} is $O(p + r\lambda n^2)$.

\begin{corollary}\cite{CX17}\label{thm:CX-min-cut}
Let $G$ be an $n$-vertex
hypergraph of size $p$ (possibly multihypergraph). Then, there exists a deterministic algorithm \textsc{CXMinCut} that takes $G$ as input and runs in time $O(p+r\lambda n^2)$ to compute the min-cut capacity $\lambda$ of $G$ (and also return a min-cut). 
\end{corollary}

The next result summarizes the random contraction based algorithm for computing global min-cut in low-rank hypergraphs. 
\begin{theorem}\cite{FPZ19}\label{thm:FPZ-min-cut} 
Let $G$ be an $r$-rank $n$-vertex hypergraph of size $p$ (possibly multihypergraph). Then, there exists a randomized algorithm \textsc{FPZMinCut}
that takes $G$ as input and runs in time $\tilde{O}(p+n^r)$ to return a global min-cut of $G$ with high probability. 
\end{theorem}

The next result summarizes another recent algorithm for computing global min-cut in hypergraphs via the isolating cuts technique.
\begin{theorem}\cite{CQ21}\label{thm:CQ-min-cut}
Let $G$ be an $r$-rank $n$-vertex hypergraph of size $p$ with $m$ hyperedges (possibly multihypergraph). Then, there exists a randomized algorithm
\textsc{CQMinCut} that takes $G$ as input and runs in time $\tilde{O}(\sqrt{pn(m+n)^{1.5}})$ to return a global min-cut of $G$ with high probability. 
\end{theorem}

\section{Structural theorem}\label{sec:structural-theorem-main}
We prove Theorem \ref{thm:number-of-hyperedges-in-mincuts} in this section. 
We call a min-cut $(C, V \setminus C)$ \emph{moderate-sized} if its size $\min\{|C|,|V \setminus C|\}$ is in the range $(r-\log{(\lambda/4r)}/\log{n}, (\lambda/2)^{1/r})$. 
In Section \ref{sec:no-medium-cuts}, we show that a hypergraph has no moderate-sized min-cuts. In Section \ref{sec:expander-decomposition} we show that every hypergraph has an expander decomposition—--i.e., a partition of the vertex set into parts which have good expansion such that the number of hyperedges intersecting multiple parts is small. In Section \ref{sec:trim-and-shave}, we define $\trim$ and $\shave$ operations and prove properties about these operations. 

We prove Theorem \ref{thm:number-of-hyperedges-in-mincuts} in Section \ref{sec:structural-theorem} as follows: We assume that there are no min-cuts of small size (i.e., of size at most $r-\log{(\lambda/4r)}/\log{n}$) and bound the number of hyperedges in the union of all min-cuts. For this, we find an expander decomposition and apply the $\trim$ and $\shave$ operations to each part of the decomposition. Since a hypergraph cannot have moderate-sized min-cuts, and there are no small-sized min-cuts by assumption, it follows that every min-cut has large size. We use this fact to show that the result of the $\trim$ and $\shave$ operations is a partition of $V$ such that (1) none of the parts intersect both sides of any min-cut and (2) the number of hyperedges crossing the parts satisfies the bound in Condition \ref{itm:few-edges} of the theorem. 

\subsection{No moderate-sized min-cuts}\label{sec:no-medium-cuts}
The following lemma is the main result of this section. It shows that there are no moderate-sized min-cuts. 

\begin{lemma}\label{cor:small-side-not-tiny-implies-small-side-large}
Let $G=(V,E)$ be an $r$-rank $n$-vertex hypergraph with connectivity $\lambda$ such that $\lambda\ge r2^{r+1}$. Let $(C, V \setminus C)$ be an arbitrary min-cut. If $\min\{|C|, |V \setminus C|\}> r - \log (\lambda/4r)/\log n$, then $\min\{|C|,|V \setminus C|\}\ge (\lambda/2)^{1/r}$.
\end{lemma}



\begin{proof}
Without loss of generality, let $|C|=\min\{|C|, |V \setminus C|\}$. Let $t:=|C|$ and $s:=r-\log{(\lambda/4r)}/\log{n}$. We know that $s< t$. 
Suppose for contradiction that $t < (\lambda/2)^{1/r}$. We will show that there exists a vertex $v$ with $|\delta(v)|<\lambda$, thus contradicting the fact that $\lambda$ is the min-cut capacity. 
We classify the hyperedges of $G$ which intersect $C$ into three types as follows: 
\begin{align*}
    E_1 &:= \{e \in E \colon e \subseteq C \}, \\
    E_2 &:= \{e \in E \colon C \subsetneq e \}, \text{ and} \\
    E_3 &:= \{e \in E \colon \emptyset \neq e \cap C \neq C \text{ and } e \cap (V \setminus C) \neq \emptyset \}.
\end{align*}
We note that $E_1$ is the set of hyperedges that are fully contained in $C$, $E_2$ is the set of hyperedges which contain $C$ as a proper subset, and $E_3$ is the set of hyperedges that intersect $C$ but are not in $E_1\cup E_2$. 
We distinguish two cases.

\noindent \textbf{Case 1:} Suppose $t < r$. Then, the number of hyperedges that can be fully contained in $C$ is at most $2^r$, so $|E_1| \leq 2^r$. Since $(C, V \setminus C)$ is a min-cut, we have that $\lambda = |\delta(C)| = |E_2| + |E_3|$. We note that the number of hyperedges of size $i$ that contain all of $C$ is at most $\binom{n-t}{i-t}$. Hence,  
\[
|E_2| \leq \sum_{i=t+1}^{r} \binom{n-t}{i-t} = \sum_{i=1}^{r-t} \binom{n-t}{i} \leq \sum_{i=1}^{r-t} n^i \leq 2n^{r-t}.
\]
Since each hyperedge in $E_3$ contains at most $t-1$ of the vertices of $C$, a uniform random vertex of $C$ is in such a hyperedge with probability at most $(t-1)/t$. Therefore, if we pick a uniform random vertex from $C$, the expected number of hyperedges from $E_3$ incident to it is at most $(\frac{t-1}{t})|E_3|$. Hence, there exists a vertex $v\in C$ such that 
\[|\delta(v) \cap E_3| \leq \left(\frac{t-1}{t}\right)|E_3| \leq \left(\frac{t-1}{t}\right)|\delta(C)| \leq \left(\frac{r-1}{r}\right)\lambda.
\]

Combining the bounds for $E_1$, $E_2$, and $E_3$, we have that
\begin{align*}
    |\delta(v)| &= |\delta(v) \cap E_1| + |E_2| + |\delta(v) \cap E_3| \\
    &\leq |E_1| + |E_2| + |\delta(v) \cap E_3| \\
    &\leq 2^r + 2n^{r-t} + \left(\frac{r-1}{r}\right)\lambda \\
    &< 2^r + 2n^{r-s} + \left(\frac{r-1}{r}\right)\lambda \quad \quad \quad \quad \text{(since $t> s$)} \\
    &= 2^r + \frac{\lambda}{2r} + \left(\frac{r-1}{r}\right)\lambda \quad \quad \text{(since $\lambda = 4rn^{r-s}$)}\\
    &= \lambda +\frac{r2^{r+1}-\lambda}{2r}\\
    &\leq \lambda. \quad \quad \quad \quad \quad \quad \text{(since $\lambda \ge r2^{r+1}$)}
\end{align*}
Consequently, $|\delta(v)| < \lambda$, contradicting the fact that $\lambda$ is the min-cut capacity. 


\noindent \textbf{Case 2:} Suppose $t \geq r$. Then, no hyperedge can contain $C$ as a proper subset, so $|E_2| = 0$. For each $v\in C$, the number of hyperedges $e$ of size $i$ such that $v\in e\subseteq C$ is at most $\binom{t-1}{i-1}$. Hence, 
\[
|\delta(v) \cap E_1| \leq \sum_{i=2}^r \binom{t-1}{i-1} = \sum_{i=1}^{r-1} \binom{t-1}{i} \leq \sum_{i=1}^{r-1} t^i  \leq 2t^{r-1}.
\]
Since each hyperedge in $E_3$ contains at most $r-1$ of the vertices of $C$, a random vertex of $C$ is in such a hyperedge with probability at most $(r-1)/t$. Therefore, if we pick a random vertex from $C$, the expected number of hyperedges from $E_3$ incident to it is at most $(\frac{r-1}{t})|E_3|$. Hence, there exists a vertex $v\in C$ such that 
\[
|\delta(v) \cap E_3| \leq \left(\frac{r-1}{t}\right)|E_3| \leq \left(\frac{r-1}{t}\right)\lambda.
\]

Since $t < (\lambda/2)^{1/r}$ and $t \geq r$, we have that $2t^r/\lambda < t-r+1$. Combining this with our bounds on $|\delta(v) \cap E_1|$ and $|\delta(v) \cap E_3|$, we have that 
\begin{align*}
    |\delta(v)| &= |\delta(v) \cap E_1| + |\delta(v) \cap E_3|
    \leq 2t^{r-1} + \left(\frac{r-1}{t}\right)\lambda 
    = \left(r-1 + \frac{2t^r}{\lambda}\right)\frac{\lambda}{t}
    < \lambda.
\end{align*}
Consequently, $|\delta(v)| < \lambda$, contradicting the fact that $\lambda$ is the min-cut capacity.
\end{proof}

\subsection{Hypergraph expander decomposition}\label{sec:expander-decomposition}
In this section, we prove the existence of an expander decomposition and also design an algorithm to construct it efficiently in constant rank hypergraphs.

\begin{lemma}[Existential hypergraph expander decomposition]\label{lem:exp-decomposition-existence}
For every $r$-rank $n$-vertex hypergraph $G=(V,E)$ with $p:=\sum_{e\in E}|e|$ and every positive real value $\phi \le 1/(r-1)$, there exists a partition $\{X_1, \ldots, X_k\}$ of the vertex set $V$ such that the following hold:
\begin{enumerate}
    \item $\sum_{i=1}^k |\delta(X_i)| = O(r\phi p\log n)$, and 
    
    \item For every $i \in [k]$ and every non-empty set $S \subset X_i$, we have that
    \[
    |E^o(S, X_i \setminus S)| \geq \phi \cdot \min\{\vol(S), \vol(X_i \setminus S)\}. 
    \]
\end{enumerate}
\end{lemma}

\begin{remark}
We note that for $\phi = 1$, the trivial partition of vertex set where each part is a single vertex satisfies both conditions.
\end{remark}
\begin{proof}
We prove this by reducing to multigraphs and then using the existence of expander decomposition in multigraphs as stated in \Cref{thm:expander-decomp-exists}.
We create a graph $G' = (V, E')$ from $G$ as follows: For each $e \in E$, pick $v \in e$ arbitrarily, and replace $e$ with the set of edges $\{\{u,v\} : u \in e \setminus \{v\} \}$. That is, we replace each hyperedge of $G$ with a star. Consequently, each hyperedge $e$ is replaced by 
$|e|-1$ edges. Therefore, $|E'| \leq \sum_{e \in E} |e| = p$.  Furthermore, for every $X \subseteq V$, we have that $\vol_G(X) \leq \vol_{G'}(X)$ and $|\delta_G(X)| \leq |\delta_{G'}(X)|$. Also, for every $S \subseteq X$, we have that $|E_{G}^o(S,X \setminus S)| \geq \frac{1}{r-1}|E_{G'}(S, X \setminus S)|$.

Now, let $\{X_1, \ldots, X_k\}$ be the partition obtained by using \Cref{thm:expander-decomp-exists} on input graph $G'$ with parameter $\phi'= (r-1)\phi$. Then, 
\[
\sum_{i=1}^k |\delta_G(X_i)| \leq \sum_{i=1}^k |\delta_{G'}(X_i)| = O(\phi'|E'|\log n) = O(r\phi p\log n). 
\]
Moreover, for every $i\in [k]$ and every non-empty set $S\subset X_i$, we have that 
\begin{align*}
|E_G^o(S, X_i \setminus S)| 
&\geq \left(\frac{1}{r-1}\right)|E_{G'}(S, X_i \setminus S)| \\
&\geq \left(\frac{\phi'}{r-1}\right)\min\{\vol_{G'}(S), \vol_{G'}(X_i \setminus S)\} \\
&\geq \phi \min\{\vol_G(S), \vol_G(X_i \setminus S)\}.
\end{align*}

\end{proof}

Next, we show that it is possible to efficiently find an expander decomposition with almost the same guarantees as in the existential theorem. 
\begin{lemma}[Algorithmic hypergraph expander decomposition]
\label{lem:efficient-expander-decomp}
There exists a deterministic algorithm that takes as input an $r$-rank $n$-vertex hypergraph $G=(V,E)$ with $p := \sum_{e \in E} |e|$ and a positive real value $\phi \leq 1/(r-1)$
and runs in time $O(p^{1+o(1)})$
to return a partition $\{X_1, \dots, X_k\}$ of the vertex set $V$ such that the following hold:
\begin{enumerate}
    \item \label{itm:few-cross-edges-xdecomp} 
    $\sum_{i=1}^k \left|\delta(X_i)\right| = O(r\phi p^{1+o(1)})$, and 
    \item For every $i \in [k]$, and for every non-empty set $S \subset X_i$, we have that $|E^o(S, X_i \setminus S)| \geq \phi \cdot \min\{\vol(S), \vol(X_i \setminus S)\}$.
\end{enumerate}
\end{lemma}

\begin{proof}
The proof is similar to that of \Cref{lem:exp-decomposition-existence}---via a reduction to graphs and using the result of \cite{cglnps20} mentioned in Theorem \ref{thm:CGLNPS-expander-decomp} (which is a constructive version of \Cref{thm:expander-decomp-exists}). We note that Theorem \ref{thm:CGLNPS-expander-decomp} gives a worse guarantee on $\sum_{i=1}^k |\delta_G(X_i)|$ than that in \Cref{thm:expander-decomp-exists}. Hence, we revisit the calculations for Condition \ref{itm:few-cross-edges-xdecomp} below. The notations that we use below are borrowed from the proof of \Cref{lem:exp-decomposition-existence}.

Let $\{X_1, \ldots, X_k\}$ be the partition returned by algorithm \textsc{GraphExpanderDecomp} (from Theorem \ref{thm:CGLNPS-expander-decomp}) on input graph $G'$ with parameter $\phi'= (r-1)\phi$. Then, 
\[
\sum_{i=1}^k |\delta_G(X_i)| \leq \sum_{i=1}^k |\delta_{G'}(X_i)| \leq  = O\left(\phi'|E'|^{1+o(1)}\right) = O\left(r\phi p^{1+o(1)}\right). 
\]

The graph $G'$ in the proof of \Cref{lem:exp-decomposition-existence} (also used in this proof) can be constructed in time $O(p)$. The algorithm \textsc{GraphExpanderDecomp} from Theorem \ref{thm:CGLNPS-expander-decomp} on input graph $G'$ runs in time $O(p^{1+o(1)})$. Hence, the total running time of this algorithm is $O(p^{1+o(1)})$.
\end{proof}

\subsection{Trim and Shave operations}\label{sec:trim-and-shave}
In this section, we define the trim and shave operations and prove certain useful properties about them. 
\begin{definition}\label{def:trim-and-shave}
Let $G=(V,E)$  be an $r$-rank hypergraph. For $X\subseteq V$, let 
\begin{enumerate}
    \item \trim$(X)$ be the set obtained by repeatedly removing from $X$ a vertex $v$ with $d_X(v) < d(v)/2r$ until no such vertices remain, 
    \item \shave$(X)$ $:= \{v \in X \colon  d_X(v) > (1-1/r^2)d(v) \}$, and 
    \item \shave$_k(X) := \shave(\shave \cdots(\shave(X)))$ be the result of applying $k$ consecutive shave operations to $X$.
\end{enumerate}
\end{definition}
We emphasize that trim is an adaptive operation while shave is a non-adaptive operation and $\shave_k(X)$ is a sequence of shave operations. We now prove certain useful properties about these operations. In the rest of this subsection, let $G=(V,E)$ be an $r$-rank $n$-vertex hypergraph with minimum degree $\delta$ and min-cut capacity $\lambda$. 
The following claim shows that the $\trim$ operation on a set $X$ that has small intersection with a min-cut further reduces the intersection. 
\begin{claim}\label{clm:trim-bound}
Let $(C, V \setminus C)$ be a min-cut. Let $X$ be a subset of $V$ and  $X':=\trim(X)$. 
If $\min\{|X \cap C|, |X \cap (V \setminus C)|\} \leq (\delta/6r^2)^{1/(r-1)}$, then 
\[
\min\{|X'\cap C|, |X'\cap (V \setminus C)|\}\le 3r^2.
\]
\end{claim}

\begin{proof}
Without loss of generality, we assume that $|X\cap C|=\min\{|X\cap C|, |X\cap (V \setminus C)|\}$. 
We have that 
\begin{align*}
    \delta &\geq \lambda \\
    &\geq \left|E(X' \cap C, X' \cap (V \setminus C))\right| \\
    &\geq \left(\frac{1}{r-1}\right)(\vol_{G[X']}(X' \cap C) - r|E(X' \cap C, X' \cap C)|) \quad \quad \text{(By Proposition \ref{prop:crossing-edges-volume-relationship})}\\
    &\geq \left(\frac{1}{r-1}\right)\left(\frac{1}{2r}\delta|X' \cap C| - r|X' \cap C|^r\right).
\end{align*}
The last inequality above is by definition of \trim. Rearranging the above, we obtain that
\[
r|X' \cap C|^r \geq \left(\frac{1}{2r}|X' \cap C| - (r-1)\right)\delta.
\]
Dividing by $|X' \cap C|$, we have that
\[
r|X' \cap C|^{r-1} \geq \left(\frac{1}{2r} - \frac{(r-1)}{|X' \cap C|}\right)\delta.
\]
We note that $|X' \cap C| \leq |X \cap C| = \min\{|X \cap C|, |X \cap (V \setminus C)|\} \leq (\delta/6r^2)^{1/(r-1)}$. Therefore,  
\[
\frac{\delta}{6r} \geq \left(\frac{1}{2r}-\frac{r-1}{|X' \cap C|}\right)\delta.
\]
Dividing by $\delta$ and rearranging, we obtain that
\[
|X' \cap C| \leq 3r(r-1) \leq 3r^2.
\]

\end{proof}

The following claim shows that the $\shave$ operation on a set $X$ which has small intersection with a large-sized min-cut further reduces the intersection. 
\begin{claim}\label{clm:shave-bound}
Suppose 
$\lambda \ge r(4r^2)^r$. 
Let $(C, V \setminus C)$ be a min-cut with 
$\min\{|C|, |V \setminus C|\} \ge 4r^2$. 
Let $X'$ be a subset of $V$ and $X'' := \shave(X')$. 
If $0 < \min\{|X' \cap C|, |X' \cap (V \setminus C)|\} \leq 3r^2$, then 
\[
\min\{|X''\cap C|, \left|X''\cap (V \setminus C)\right|\}\le \min\{\left|X'\cap C\right|, \left|X'\cap (V \setminus C)\right|\}-1.
\]
\end{claim}

\begin{proof}
Without loss of generality, we assume that $|X' \cap C| = \min\{|X' \cap C|, |X' \cap (V \setminus C)|\}$. Since $X'' \subseteq X'$, we have that $|X'' \cap C| \leq |X' \cap C|$. Thus, we only need to show that this inequality is strict. Suppose for contradiction that $|X''\cap C|=|X'\cap C|$. We note that $0 < |X'' \cap C| \leq 3r^2$.

Let $Z := X' \cap C = X'' \cap C$, and let $C' := C - X'$. 
Since $|C|\ge \min\{|C|, |V \setminus C|\} \ge 4r^2$ and $|Z| \leq 3r^2$, we know that $C'$ is nonempty. 

We note that $Z\subseteq X''$. By definition of $\shave$, we have that
\[
\vol_{X'}(Z) = \sum_{v \in Z} d_{X'}(v) > \sum_{v \in Z} \left(1-\frac{1}{r^2}\right)d(v) = \left(1-\frac{1}{r^2}\right)\vol(Z).
\]

We note that $|E(Z, V \setminus C)| \geq |E(Z, X' \setminus C)| = |E(Z, X' \setminus Z)|$, so by Proposition \ref{prop:crossing-edges-volume-relationship}, we have that
\begin{align*}
\left|E(Z, V \setminus C)\right| &\geq \left|E(Z, X' \setminus Z)\right| \\ &\geq \left(\frac{1}{r-1}\right)(\vol_{X'}(Z) - r|E[Z]|) \\ &> \left(\frac{1}{r-1}\right)\left(\left(1-\frac{1}{r^2}\right)\vol(Z)-r|Z|^r \right).
\end{align*}

We also know from the definition of \shave\ that
\[
|E(Z, C \setminus Z)| \leq \sum_{v \in Z} \left|E(\{v\}, C \setminus Z)\right| \leq \sum_{v \in Z} \frac{1}{r^2}d(v) = \frac{\vol(Z)}{r^2}.
\]

Therefore, using our assumption that 
$\lambda \ge r(4r^2)^r$, 
we have that
\begin{align*}
    \left|E(Z, (V \setminus C))\right| &> \left(\frac{1}{r-1}\right)\left(\left(1-\frac{1}{r^2}\right)\vol(Z) - r|Z|^r \right) \\
    &= \left(\frac{r^2-1}{r^2(r-1)}\right)\vol(Z) - \left(\frac{r}{r-1}\right)|Z|^r \\
    &= \frac{\vol(Z)}{r^2} + \frac{\vol(Z)}{r} - \left(\frac{r}{r-1}\right)|Z|^r \\
    &= \frac{\vol(Z)}{r^2} + \frac{\sum_{v \in Z} d(v)}{r} - r|Z|^r \\
    &\geq \frac{\vol(Z)}{r^2} + \frac{|Z|\lambda}{r} - r|Z|^r \\
    &\geq \frac{\vol(Z)}{r^2} + (4r^2)^r|Z| - r|Z|^r\\
    &\geq \frac{\vol(Z)}{r^2} + (4r^2)|Z|^r - r|Z|^r \\
    &\geq \frac{\vol(Z)}{r^2} \\
    &\geq |E(Z, C \setminus Z)|.
\end{align*}
We note that $E(Z, (V \setminus C))$ is the set of hyperedges which are cut by $C$ but not $C'$, while $E(Z, C \setminus Z)$ is the set of hyperedges which are cut by $C'$ but not $C$. Since we have shown that 
$|E(Z, V \setminus C)| > |E(Z, C \setminus Z)|$, 
we conclude that $|\delta(C)| > |\delta(C')|$. Since $(C, V\setminus C)$ is a min-cut and $\emptyset \neq C'\subseteq C\subsetneq V$, this is a contradiction.
\end{proof}

The following claim shows that the $\trim$ operation increases the cut capacity by at most a constant factor. 
\begin{claim}\label{clm:hyperedges-lost-to-trim}
Let $X$ be a subset of $V$ and $X' := \trim(X)$. Then, 
\begin{enumerate} 
\item $|E[X] - E[X']| \leq |\delta(X)|$, and 
\item $|\delta(X')| \leq 2|\delta(X)|$.
\end{enumerate}
\end{claim}
\begin{proof}
We will prove the first part of the claim. The second part follows from the first part. 

For a set $S \subseteq V$, we define a potential function $f(S) := \sum_{e \in \delta(S)} |e \cap S|$. Since every hyperedge in $\delta(S)$ must include 
at least $1$ and at most $r-1$ 
vertices of $S$, we have that $|\delta(S)| \leq f(S) \leq (r-1)|\delta(S)|$. We note that the removal of a vertex $v$ from $X$ during the \trim\ process decreases $f(X)$ by at least $(1-\frac{1}{2r})d(v) - \frac{(r-1)}{2r}d(v) = \frac{1}{2}d(v)$. This is because, from the definition of \trim, at least a $1-\frac{1}{2r}$ fraction of $v$'s hyperedges were part of $\delta(X)$ before $v$ was removed, and each of these hyperedges losing a vertex causes $f(X)$ to decrease by one, while $v$'s removal can add at most $\frac{1}{2r}d(v)$ new hyperedges to $\delta(X)$, and each of these hyperedges can contribute at most $r-1$ to $f(X)$. We also note that when $v$ is removed, the cut capacity $|\delta(X)|$ increases by at most $|E(\{v\}, X \setminus \{v\})| \leq \frac{1}{2r}d(v)$. 
Therefore, the removal of a vertex $v$ from $X$ during the \trim\ process increases the cut capacity $|\delta(X)|$ by at most $\frac{1}{2r}d(v)$ and decreases the potential function $f(X)$ by at least $\frac{1}{2}d(v)$. Thus, for every additional hyperedge added to $\delta(X)$, the potential function must decrease by at least $\frac{d(v)/2}{d(v)/2r} = r$. Before the $\trim$ operation, we have $f(X) \leq (r-1)|\delta(X)|$, so the number of hyperedges that can be removed from $X$ at the end of the $\trim$ operation is at most $\frac{1}{r}(r-1)|\delta(X)| \leq |\delta(X)|$. 
\end{proof}

The following claim shows that the $\shave$ operation increases the cut capacity by at most a factor of $r^3$. 
\begin{claim}\label{clm:hyperedges-lost-to-shave}
Let $X'$ be a subset of $V$ and $X'' := \shave(X')$. Then
\begin{enumerate}
    \item $|E[X'] - E[X'']| \leq r^2(r-1)|\delta(X')|$, and
    \item $|\delta(X'')| \leq r^3|\delta(X')|$.
\end{enumerate} 
\end{claim}

\begin{proof}
We will prove the first part of the claim. The second part follows from the first part. 

We note that the removal of a vertex $v$ from $X'$ during the $\shave$ process decreases the number of hyperedges fully contained in $X$ by at most $d(v)$. We also know that if $v$ is removed, we must have $|\delta(v) \cap \delta(X')| \geq d(v)/r^2$, and therefore $d(v) \leq r^2|\delta(v) \cap \delta(X')|$. Thus, the total number of hyperedges removed from $X'$ in a single $\shave$ operation is at most $r^2 \sum_{v \in X'} |\delta(v) \cap \delta(X')| \leq r^2(r-1)|\delta(X')|$.
\end{proof}

\subsection{Proof of Theorem \ref{thm:number-of-hyperedges-in-mincuts}}\label{sec:structural-theorem}
In this section, we restate and prove Theorem \ref{thm:number-of-hyperedges-in-mincuts}.
\numEdgesMincuts*
\begin{proof}
Suppose the first conclusion does not hold. Then, by Lemma \ref{cor:small-side-not-tiny-implies-small-side-large}, the smaller side of every min-cut has size at least $(\lambda/2)^{1/r} \ge 4r^2$.  Let $(C, V \setminus C)$ be an arbitrary min-cut. We use Lemma \ref{lem:exp-decomposition-existence} with $\phi = (6r^2/\lambda)^{1/(r-1)}$ to get an expander decomposition $\mathcal{X}=\{X_1, \ldots, X_k\}$. We note that $\phi\le 1/(r-1)$ holds by the assumption that $\lambda \ge r(4r^2)^r$. For $i \in \{1, \dots, k\}$,  let $X'_i := \trim(X_i)$ and $X''_i := \shave_{3r^2}(X'_i)$. 

Let $i\in [k]$. By the definition of the expander decomposition and our choice of $\phi = (6r^2/\lambda)^{1/(r-1)}$, we have that
\begin{align*}
\lambda &\geq |E^o(X_i \cap C, X_i \cap (V \setminus C))| \\
&\geq \left(\frac{6r^2}{\lambda}\right)^{\frac{1}{r-1}} \min \{\vol(X_i \cap C), \vol(X_i \setminus C)\} \\
&\geq \left(\frac{6r^2}{\lambda}\right)^{\frac{1}{r-1}} \delta \min \{|X_i \cap C|, |X_i \setminus C|\}.
\end{align*}
Thus, $\min\{|X_i \cap C|, |X_i \setminus C| \} \leq (\lambda/\delta) (\lambda/6r^2)^{1/(r-1)} \leq (\lambda/6r^2)^{1/(r-1)} \leq (\delta/6r^2)^{1/(r-1)}$.

Therefore, by Claim \ref{clm:trim-bound}, we have that $\min\{|X'_i \cap C|, |X'_i \cap (V \setminus C)|\} \leq 3r^2$. 
We recall that $\lambda \ge r(4r^2)^r$ and every min-cut has size at least $4r^2$. 
By $3r^2$ repeated applications of Claim \ref{clm:shave-bound}, we have that $\min\{|X''_i \cap C|, |X''_i \cap (V \setminus C)|\} = 0$.

Let $\mathcal{X}'' := \{X''_1, \dots, X''_k\}$. 
Since $\min\{|X''_i \cap C|, |X''_i \cap (V \setminus C)|\} = 0$ for every min-cut $(C, V\setminus C)$ and every $X''_i\in \mathcal{X}''$, it follows 
that no hyperedge crossing a min-cut is fully contained within a single part of $\mathcal{X}''$. 
Thus, it suffices to show that $|E - \bigcup_{i=1}^k E[X''_i]|$ is small---i.e., the number of hyperedges not contained in any of the parts of $\mathcal{X}''$ is $\tilde{O}_r(m/\lambda^{\frac{1}{r-1}})$. 

By Claim \ref{clm:hyperedges-lost-to-trim}, we have that $|E[X_i] - E[X_i']| \leq 2|\delta(X_i)|$ and $|\delta(X_i')| \leq 2|\delta(X_i)|$ for each $i\in [k]$. 
By the second part of Claim \ref{clm:hyperedges-lost-to-shave}, we have that $|\delta(\shave_{j+1}(X'_i))| \leq r^3|\delta(\shave_{j}(X'_i))|$ for every non-negative integer $j$. Therefore, by repeated application of the second part of Claim \ref{clm:hyperedges-lost-to-shave}, for every $j \in [3r^2]$, we have that $|\delta(\shave_j(X'_i))| \leq 2r^{3j}|\delta(X_i)|$. By the first part of Claim \ref{clm:hyperedges-lost-to-shave}, for every $j \in [3r^2]$, we have that $|E[\shave_{j-1}(X_i')] - E[\shave_{j}(X_i')]| \leq r^3|\delta(\shave_{j-1}(X'_i))| \leq  2r^{3j}|\delta(X_i)|$. 
Therefore,
\begin{align*} 
    \left|E - \bigcup_{i=1}^k E[X''_i]\right| 
    &\le \left|E - \bigcup_{i=1}^k E[X_i]\right| + \sum_{i=1}^k \left|E[X_i] - E[X''_i]\right| \\ 
&= \left|E - \bigcup_{i=1}^k E[X_i]\right| + \sum_{i=1}^k \left( |E[X_i] - E[X'_i]| + \sum_{j=1}^{3r^2} |E[\shave_{j-1}(X_i')] - E[\shave_{j}(X_i')]| \right) \\
&\leq \left|E - \bigcup_{i=1}^k E[X_i]\right| + \sum_{i=1}^k \left( 2|\delta(X_i)| + \sum_{j=1}^{3r^2} 2r^{3j}|\delta(X_i)| \right) \\
&= \left|E - \bigcup_{i=1}^k E[X_i]\right| + \sum_{i=1}^k |\delta(X_i)|\left(2 + \sum_{j=1}^{3r^2} 2r^{3j} \right) \\
&\leq \left|E - \bigcup_{i=1}^k E[X_i]\right| + \sum_{i=1}^k 3r^{9r^2}|\delta(X_i)| \\
&\leq \left|E - \bigcup_{i=1}^k E[X_i]\right| + 3r^{9r^2}\sum_{i=1}^k \left|E(X_i, V \setminus X_i)\right| \\
&\leq 4r^{9r^2}\sum_{i=1}^k \left|E(X_i, V \setminus X_i)\right|.
\end{align*}

By Lemma \ref{lem:exp-decomposition-existence}, since $\mathcal{X}$ is an expander decomposition for $\phi = (6r^2/\lambda)^{1/(r-1)}$ and since $p = \sum_{e \in E} |e| \leq mr$, we have that 
\[
\sum_{i=1}^k \left|E(X_i, V \setminus X_i)\right| 
= O(r\phi p \log n) 
= O(r) \left(\frac{6r^2}{\lambda}\right)^{\frac{1}{r-1}}p \log n 
= O(r^2)  \left(\frac{6r^2}{\lambda}\right)^{\frac{1}{r-1}}m\log n.
\]
Thus, $|E - \bigcup_{i=1}^k E[X_i'']| = O(r^{9r^2+2}(6r^2 / \lambda)^{1/(r-1)} m\log n)$, thus proving the second conclusion.
\end{proof}

\section{Algorithm}\label{sec:min-cut-algo}

In this section, we prove Theorem \ref{theorem:algo}. In Section \ref{sec:small-size-min-cut}, we give an algorithm to find a min-cut in hypergraphs where some min-cut has small size. In Section \ref{sec:large-min-cut}, we give an algorithm which uses the expander decomposition and the $\trim$ and $\shave$ operations defined in Section \ref{sec:trim-and-shave}  to find a min-cut if the size of every min-cut is large. In Section \ref{sec:main-algorithm}, we apply the sparsification algorithm of Chekuri and Xu, followed by a combination of our two algorithms in order to prove Theorem \ref{theorem:algo}.

\subsection{Small-sized Hypergraph Min-Cut}\label{sec:small-size-min-cut}
We recall that the size of a cut is the number of vertices on the smaller side of the cut. In this section, we give a randomized algorithm to find connectivity in hypergraphs in which some min-cut has size at most $s$. 
The following is the main result of this section. 

\begin{theorem}\label{thm:new-small-mincut-algo}
Let $s$ be a positive integer and let $G=(V,E)$ be an $r$-rank hypergraph of size $p$ that has a min-cut $(C,V \setminus C)$ with $|C| \leq s$. Then, there exists a randomized algorithm $\textsc{SmallSizeMinCut}$ which takes $G$ and $s$ as input and runs in time $\tilde{O}_r(s^{6r}p)$ to return a min-cut of $G$ with high probability.
\end{theorem}

\begin{remark}
We note that the above result gives a randomized algorithm. In contrast, we give a simple but slightly slower deterministic algorithm for this problem in Appendix \ref{sec:det-small-size-min-cut}. 
\end{remark}

We prove Theorem \ref{thm:new-small-mincut-algo} using one of two possible algorithms based on the connectivity $\lambda$ of the input hypergraph: By Theorem \ref{thm:CX-approximation}, we can find a $k$ such that $\lambda \leq k \leq 3\lambda$ in $O(p)$ time. If $k \leq 2700s^r\log n$, then we use the algorithm from Theorem \ref{thm:small-lambda-small-cut} to find a min-cut in time $\tilde{O}_r(s^{3r}p + s^{6r})$ and if $k > 2700s^r \log n$ (i.e., $\lambda > 900s^r \log n$), then we use the algorithm from Theorem \ref{thm:big-lambda-small-cut} to find a min-cut in time $\tilde{O}_r(s^{3r}p)$.
These two algorithms exploit two slightly different ways of representing the hypergraph as a graph.

\subsubsection{Algorithm for small $\lambda$}
This section is devoted to proving the following theorem, which allows us to find a min-cut in a hypergraph with a small-sized min-cut. The algorithm is fast if the min-cut capacity is sufficiently small.

\begin{theorem}\label{thm:small-lambda-small-cut}
Let $s$ be a positive integer and let $G=(V,E)$ be an $r$-rank hypergraph of size $p$ that has a min-cut $(C,V \setminus C)$ with $|C| \leq s$. Then, there exists a randomized algorithm which takes $G$ and $s$ as input and runs in time $\tilde{O}_r(p\lambda^2rs^r + \lambda^4r^2s^{2r})$ to return a min-cut of $G$ with high probability.
\end{theorem}

In the rest of this subsection, let $s$ be a positive integer and let $G=(V, E)$ be an $r$-rank $n$-vertex hypergraph of with $m$ hyperedges and size $p$ that has a min-cut $(C, V\setminus C)$ with $|C|\le s$. Let $\lambda$ be the min-cut capacity of $G$ and let $k$ be the approximation of $\lambda$ obtained by running the algorithm from Theorem \ref{thm:CX-approximation}. If $p \leq 512k^2rs^r$, then $n \leq p \leq 512k^2rs^r$. In this case, we run the algorithm \textsc{SlowMinCut} from Theorem \ref{thm:slow-min-cut} to compute the min-cut in $G$ in time $O(np) = \tilde{O}(\lambda^4r^2s^{2r})$. Henceforth, we assume that $p > 512k^2rs^r$.

We construct an arc-weighted directed graph $D_G = (V_D, E_D)$ as follows: 
The vertex and arc sets of $D_G$ are given by $V_D := U_V \cup U_{E_{in}} \cup U_{E_{out}}$ and $E_D := E_1 \cup E_2$ respectively, where $U_V, U_{E_{in}}, U_{E_{out}}, E_1,$ and $E_2$ are defined next. 
The vertex set $U_V$ has a vertex $u_v$ for each $v \in V$. The vertex set $U_{E_{in}}$ has a vertex $e_{in}$ for every hyperedge $e \in E$, and the vertex set $U_{E_{out}}$ has a vertex $e_{out}$ for every hyperedge $e \in E$. The arc set $E_1$ has an arc $(e_{in}, e_{out})$ for every $e \in E$. The arc set $E_2$ has arcs $(v, e_{in})$ and $(e_{out}, v)$ for every $e \in E$ and $v \in e$. We assign a weight of $1$ to each arc in $E_1$ and a weight of $\infty$ to each arc in $E_2$. We note that $|V_D| = n + 2m$, and $|E_D| = m + 2p$. Furthermore, $D_G$ can be constructed from $G$ in $\tilde{O}(p)$ time.

\begin{definition}
A directed cut in $D_G$ is an ordered partition of the vertices into two parts $(C, V_D \setminus C)$ such that $\emptyset \neq C \cap U_V \neq U_V$, and its weight is the total weight of arcs outgoing from $C$.
\end{definition}

The following proposition relates min-cut in $G$ to min-weight directed cut in $D_G$.

\begin{proposition}\label{prop:reduction-to-directed-graph}
An algorithm that runs in time $\tilde{O}(f(p,s,r,\lambda))$ to find a min-weight directed cut in $D_G$, for some function $f$, can be used to find a min-cut in $G$ in time $\tilde{O}(p + f(p,s,r,\lambda))$.
\end{proposition}
\begin{proof}
We first note that every cut $(C, V \setminus C)$ in $G$ has a corresponding cut in $D_G$, namely $(C', V_D \setminus C')$ where $C' := \{u_v \colon v \in C\} \cup \{e_{in} \colon e \cap C \neq \emptyset \} \cup \{e_{out} \colon e \subseteq C\}$. The arcs cut by $(C', V_D \setminus C')$ will be arcs of the form $(e_{in}, e_{out})$ where $e \cap C \neq \emptyset$, but $e$ is not contained in $C$, and the number of such arcs is exactly the number of hyperedges in $\delta(C)$. Therefore, the min-weight of a directed cut in $D_G$ is at most the min-cut capacity in $G$.

Next, let $(C', V_D \setminus C')$ be a min-weight directed cut in $D_G$. Let $C := \{v \colon u_v \in C'\}$. We will show that $C' \cap U_{E_{in}} = \{e_{in} \colon e \cap C \neq \emptyset\}$ and $C' \cap U_{E_{out}} = \{e_{out} \colon e \subseteq C \}$. We note that a min-cut in $D_G$ does not cut any infinite cost hyperedges. Therefore, $C'$ contains every vertex $e_{in}$ such that $e \cap C \neq \emptyset$. Similarly, $C'$ can contain a vertex $e_{out}$ only if $e \subseteq C$. Thus, $(C', V_D \setminus C')$ cuts all arcs $(e_{in}, e_{out})$ where $e \in \delta(C)$. Therefore, the min-weight of a directed cut in $D_G$ is at least the min-cut capacity in $G$. 

The above arguments imply that the min-weight of a directed cut in $D_G$ is equal to the min-cut capacity in $G$ and moreover, given a min-weight directed cut in $D_G$, we can recover a min-cut in $G$ immediately. The run-time guarantees follow. 

\end{proof}

By Proposition \ref{prop:reduction-to-directed-graph}, in order to prove Theorem \ref{thm:small-lambda-small-cut}, it suffices to design an algorithm that runs in time $\tilde{O}_r(p\lambda^2rs^r + \lambda^4r^2s^{2r})$ to find a min-weight directed cut $(C, V_D \setminus C)$ in the directed graph $D_G$ corresponding to $G$, under the assumptions that $|C \cap U_V| \leq s$ and $p\ge 512k^2rs^r$. In Algorithm \ref{alg:s-sized-small-lambda} we give a randomized algorithm to find a min-weight directed cut under these assumptions.
Algorithm \ref{alg:s-sized-small-lambda} is inspired by an algorithm of Forster et al. \cite{FNYSY20}. We show in Theorem \ref{thm:small-size-small-min-cut-correctness} that if $x \in C$ for some min-weight directed cut $(C, V_D \setminus C)$ in $D_G$ with $|C \cap U_V| \leq s$, then running Algorithm \ref{alg:s-sized-small-lambda} on $(D_G,x,k,s)$ will return a min-weight directed cut with constant probability. Thus, running Algorithm \ref{alg:s-sized-small-lambda} on $(D_G,x,k,s)$ for each $x \in U_V$ and returning the best cut found will indeed return a min-weight directed cut in $D_G$ with constant probability; repeating this process $\log n$ times and taking the best cut found will return a min-weight directed cut with high probability. We also show in Theorem \ref{thm:small-size-small-min-cut-correctness} that Algorithm \ref{alg:s-sized-small-lambda} can be implemented to run in time $O(\lambda^3rs^r)$. Thus, running Algorithm \ref{alg:s-sized-small-lambda} $O(\log n)$ times for each $x \in U_V$ and returning the best cut found takes time $\tilde{O}(nk^3rs^r) = \tilde{O}(pk^2rs^r)$. Adding in the time for the base case where $p \leq 512k^2rs^r$, we get a total running time of $\tilde{O}(p\lambda^2rs^r + \lambda^4r^2s^{2r})$.

\begin{algorithm}[ht]
\caption{{SmallSizeSmallMinCut}($D_G,x,k,s$)}\label{alg:s-sized-small-lambda}
\begin{algorithmic}
\State Set all arcs as unmarked
\For{$i$ from $1$ to $k + 1$}
\State $y \gets NIL$
\State Run a new iteration of BFS from $x$ to get a BFS tree $T$ as follows:
\While{the BFS algorithm still has an arc $(u,v)$ to explore}
\If{$(u,v)$ is not marked}
\State Mark $(u,v)$
\If{at least $512k^2rs^r$ arcs are marked}
\State \Return $\bot$
\EndIf
\State With probability $\frac{1}{8r(4s^r+k)}$, stop BFS and set $y \gets v$
\EndIf
\EndWhile
\If{$y = NIL$}
\State \Return $V(T)$
\EndIf
\State Flip the orientation of all arcs on the $x-y$ path in $T$
\EndFor
\State \Return $\bot$
\end{algorithmic}
\end{algorithm}

Theorem \ref{thm:small-size-small-min-cut-correctness} completes the proof of theorem \ref{thm:small-lambda-small-cut}.

\begin{theorem}\label{thm:small-size-small-min-cut-correctness}
Let $(C, V_D \setminus C)$ be a min-weight directed cut in $D_G$ with $|C \cap U_V| \leq s$ and $p\ge 512k^2rs^r$. For $x\in C$, Algorithm \ref{alg:s-sized-small-lambda} returns a min-weight directed cut in $D_G$ with probability at least $\frac{3}{4}$. Moreover, Algorithm \ref{alg:s-sized-small-lambda} can be implemented to run in time $O(\lambda^3rs^r)$. 
\end{theorem}
\begin{proof}
We begin by showing the following claim.
\begin{claim}\label{clm:alg-returns-small-cut}
If Algorithm \ref{alg:s-sized-small-lambda} returns $V(T)$ in iteration $j$ of its for-loop, then $(V(T),V_D \setminus V(T))$ is a directed cut in $D_G$ with weight at most $j-1$. 
\end{claim}
\begin{proof} We note that whenever we flip the arcs in an $x-y$ path, the weight of any directed cut $(C, V_D \setminus C)$ with $x \in C$ either decreases by $1$ (if $y \in V_D \setminus C$) or stays the same. We also note that if the algorithm eventually returns $V(T)$, then
the last iteration must have explored all arcs reachable from $x$ without exploring $512k^2rs^r$ arcs. Since $|E_D| \geq p > 512k^2rs^r$, this means that not all arcs of $D_G$ are reachable from $x$, and therefore $(V(T), V_D \setminus V(T))$ is a directed cut in $D_G$. The weight of the cut $(V(T), V_D \setminus V(T))$ has decreased by at most $1$ in each previous iteration of BFS run, since this is the only time when the algorithm flips the arcs of a path. Therefore, if $V(T)$ is returned in iteration $j$ of the algorithm's for-loop, the corresponding cut in $D_G$ has weight at most $j-1$.
\end{proof}

By Claim \ref{clm:alg-returns-small-cut}, it suffices to show that with constant probability, Algorithm \ref{alg:s-sized-small-lambda} returns $V(T)$ after $\lambda +1$ iterations of BFS. By Claim \ref{clm:alg-returns-small-cut}, the algorithm cannot return $V(T)$ within $\lambda$ iterations of BFS (since this would imply that there exists a directed cut with weight less than $\lambda$). Therefore, since $\lambda \leq k$, the algorithm will always run at least $\lambda + 1$ iterations of BFS if it does not terminate because of marking $512k^2rs^r$ arcs. 
If $x \in C$, then whenever the algorithm flips the orientations of all arcs of an $x-y$ path where $y \in V_D \setminus C$, the weight of the cut $(C, V_D \setminus C)$ in the resulting graph decreases by $1$. Thus, if $x \in C$ and the algorithm fails to return a min-cut after $\lambda+1$ iterations, either the first $\lambda+1$ iterations marked more than $512k^2rs^r$ arcs, or one of the choices for vertex $y$ chosen within the first $\lambda$ iterations is in $C$.
\begin{claim}\label{clm:failure-prob-too-many-marks}
The probability that Algorithm \ref{alg:s-sized-small-lambda} returns $\bot$ due to marking at least $512k^2rs^r$ arcs within the first $\lambda + 1$ iterations of BFS is at most $1/8$.
\end{claim}
\begin{proof}
The expected number of arcs marked by a single iteration of BFS is $8r(4s^r+k)$, so the expected number of arcs marked by $\lambda+1$ iterations of BFS is \[8(\lambda+1)r(4s^r+k) \leq 8(k+1)r(4s^r+k) \leq 64rk^2s^r.\] Therefore, by Markov's inequality, the probability that at least $512k^2rs^r$ arcs are marked in the first $\lambda+1$ iterations is at most $1/8$.
\end{proof}

\begin{claim}\label{clm:failure-prob-bad-y}
The probability that Algorithm \ref{alg:s-sized-small-lambda} picks a vertex from $C$ as $y$ in some iteration of its execution is at most $1/8$.
\end{claim}

\begin{proof}
If the algorithm never terminates a BFS iteration on an arc with an endpoint in $C$, then it will always choose the vertex $y$ from $V_D \setminus C$. We recall that $|C \cap U_V| \leq s$. Therefore, $|C \cap U_{e_{out}}| \leq 2s^r$, since the number of hyperedges of size at most $r$ over a subset of vertices of size $s$ is at most 
\[
\binom{s}{r} + \binom{s}{r-1} + \binom{s}{r-2} + \dots + \binom{s}{2} \leq s^r + \dots + s^2 \leq 2s^r.
\]
Since $C$ is a min-weight directed cut, we also have that $|C \cap U_{e_{in}}| \leq |C \cap U_{e_{out}}| + \lambda \leq |C \cap U_{e_{out}}| + k$. Also, since $C$ is a min-weight directed cut, we know that for any vertex $v$ in $C \cap U_V$, all of $v$'s neighbors are also in $C$ (since the arcs between vertices in $U_V$ and $V_D \setminus U_V$ have infinite weight). Thus, every arc which is incident to a vertex in $C$ is incident to a vertex in $C \cap (U_{e_{in}} \cup U_{e_{out}})$. Since each vertex in $U_{e_{in}} \cup U_{e_{out}}$ has degree at most $r$, this means that the total number of arcs which touch vertices in $C$ is at most $r(|C \cap U_{e_{in}}| + |C \cap U_{e_{out}}|) \leq r(4s^r+k)$. Each of these arcs is marked at most once during the execution of the algorithm, and has a $1/(8r(4s^r+k))$ chance of causing BFS to terminate when it is marked. Thus, by union bound, the probability that any of these arcs cause the BFS to terminate (and potentially not flip an arc from the min-cut) is at most $1/8$. \end{proof} 

Combining Claim \ref{clm:failure-prob-too-many-marks} with Claim \ref{clm:failure-prob-bad-y}, 
and by union bound, the total probability that the algorithm returns an incorrect answer is at most $1/4$.

Finally, we analyze the runtime of Algorithm \ref{alg:s-sized-small-lambda}. The runtime is dominated by the breadth first searches. There are at most $k+1$ searches, and each breadth first search visits at most $512k^2rs^r$ arcs, so the runtime of 
Algorithm \ref{alg:s-sized-small-lambda} is $O(\lambda^3rs^r)$. 
\end{proof}



\subsubsection{Algorithm for large $\lambda$}
In this section, we prove the following theorem which allows us to find min-cut in hypergraphs where some min-cut has small size. The algorithm is fast if the min-cut capacity is sufficiently large.

\begin{theorem}\label{thm:big-lambda-small-cut}
Let $s$ be a positive integer, and let $G=(V,E)$ be an $n$-vertex $r$-rank hypergraph of size $p$ with $n \geq 2r$ that has a min-cut $(C,V \setminus C)$ with $|C| \leq s$ and min-cut capacity $\lambda \geq 900s^r\log n$. Then, there exists a randomized algorithm which takes $G$ and $s$ as input and runs in time $\tilde{O}_r(s^{3r}p)$ to return a min-cut of $G$ with high probability.
\end{theorem}

In the rest of this subsection, let $s$ be a positive integer and let $G=(V, E)$ be an $r$-rank $n$-vertex hypergraph of with $m$ hyperedges and size $p$ that has a min-cut $(C, V\setminus C)$ with $|C|\le s$ and min-cut capacity $\lambda \geq 900s^r \log n$.
We construct a vertex-weighted bipartite graph $B=(V_B, E_B)$ as follows:
The vertex set of $B$ is given by $V_B := U_V \cup U_E$, and the edge set of $B$ is $E_B$, where $U_V, U_E, E_B$ are defined next.
The vertex set $U_V$ has a vertex $u_v$ for every vertex $v \in V$.
The vertex set $U_E$ has a vertex $u_e$ for every hyperedge $e \in E$.
The edge set $E_B$ has an edge $\{u_v, u_e\}$ for every $v \in V$, $e \in E$ such that the hyperedge $e$ contains the vertex $v$ in $G$. We assign a weight of $\infty$ to each vertex in $U_V$ and a weight of $1$ to each vertex in $U_E$. We note that $|V_B| = m + n$ and $|E_B| = p$. Furthermore, the graph $B$ can be constructed from $G$ in $\tilde{O}(p)$ time.

We call a $3$-partition $(L,S,R)$ of the vertices of $B$ a {\em separator} if removing $S$ from $B$ disconnects $L$ from $R$. We define the weight of the separator as the weight of $S$. 
Throughout this section, for a vertex $u \in V_B$, we use $N(u)$ to denote $\{v \in U_V \cup U_E \colon \{u,v\} \in E_B \}$ and $N[u]$ to denote $N(u) \cup \{u\}$. 
By the weights that we have assigned to vertices of $B$, we have the following observation:
\begin{observation}\label{obs:subset-of-ue}
Every min-weight separator $(L, S, R)$ in $B$ has $S\subseteq U_E$. 
\end{observation}

The following proposition relates min-cut in $G$ to minimum weight separator in $B$ and translates our assumption that the min-cut $(C, V \setminus C)$ in $G$ has $|C| \leq s$ into a constraint on a min-weight separator in $B$. 

\begin{proposition}\label{prop:corresponence-between-cuts-and-separators}
The following hold: 
\begin{enumerate}
\item If $(L,S,R)$ is a min-weight separator in $B$, then $(\{v \colon u_v \in L\}, \{v \colon u_v \in R\})$ is a min-cut in $G$. 
\item If $(C, V \setminus C)$ is a min-cut in $G$,  
then 
for 
\begin{align*}
L &:= \{u_v \colon v \in C\} \cup \{u_e \colon e \in E \text{ and } e \subseteq C\}, \\
S &:= \{u_e \colon e \in \delta(C)\}, \text{ and}\\ 
R &:= \{u_v \colon v \in V \setminus C\} \cup \{u_e \colon e \in E \text{ and } e \subseteq V \setminus C \},
\end{align*}
the tuple $(L,S,R)$ is a min-weight separator in $B$; moreover, if $|C|\le s$, then $|L|\le 3s^r$. 
\end{enumerate}
Hence, an algorithm that runs in time $\tilde{O}(f(p,s,r,\lambda))$ to find a min-weight separator in $B$, for some function $f$, can be used to find a min-cut in $G$ in time $\tilde{O}(p + f(p,s,r,\lambda))$.
\end{proposition}

\begin{proof}
Let $(L,S,R)$ be a min-weight separator in $B$. By Observation \ref{obs:subset-of-ue} we have that $S \subseteq U_E$, and therefore $(\{v \colon u_v \in L\}, \{v \colon u_v \in R\})$ is a partition of $V$. Observation \ref{obs:subset-of-ue} also implies that $L \cap U_V$ and $R \cap U_V$ are both nonempty, and therefore $(\{v \colon u_v \in L\}, \{v \colon u_v \in R\})$ is a cut. Every edge $e$ crossing this cut is of the form $e = \{u_x,u_y\}$, where $u_x \in L$ and $u_y \in R$. For such an edge $e$, the vertex $u_e$ is adjacent to vertices in both $L$ and $R$, and therefore, by the definition of a separator, we have that $u_e \in S$. Thus, we have that the weight of the cut $(\{v \colon u_v \in L\}, \{v \colon u_v \in R\})$ is at most $|S|$.

Next, let $(C, V \setminus C)$ be a min-cut in $G$, and let the sets $L,S,R$ be as defined in the proposition. We note that by the definition of $B$, removing $S$ will disconnect $L$ from $R$. Thus, $(L,S,R)$ is indeed a separator. Furthermore, since $S \subseteq U_E$, the weight of $(L,S,R)$ is $|S| = |\delta(C)|$.

The arguments in the previous two paragraphs imply that the min-cut capacity in $G$ is at most the weight of a min-weight separator in $B$, and the weight of a minimum weight separator in $B$ is at most the min-cut capacity in $G$. Thus, the min-cut capacity in $G$ is equal to the weight of a min-weight separator in $B$. 

Next, we show that if $|C|\le s$, then $|L|\le 3s^r$. We note that $L = \{u_v \colon v \in C\} \cup \{u_e \colon e \in E \text{ and } e \subseteq C\}$. We have that $|\{u_v \colon v \in C\}| \leq s$, and \[|\{u_e \colon e \in E \text{ and } e \subseteq C\}| \leq \binom{s}{r} + \binom{s}{r-1} \dots \binom{s}{2} \leq s^r + \dots s^2 \leq 2s^r.\]
Thus, $|L| \leq s + 2s^r \leq 3s^r$.

The run-time to obtain a min-cut of $G$ from a min-weight separator of $B$ follows from the fact that $B$ can be constructed in $\tilde{O}(p)$ time. 
\end{proof}

By 
Proposition \ref{prop:corresponence-between-cuts-and-separators}, 
in order to prove Theorem \ref{thm:big-lambda-small-cut}, it suffices to design an algorithm  that runs in time $\tilde{O}_r(s^{3r}p)$ to find a min-weight separator $(L,S,R)$ in the bipartite graph $B$, under the assumption that there exists a min-weight separator $(L, S, R)$ in $B$ with $|L| \leq 3s^r$. To achieve this, we use a slightly modified version of on algorithm of Li et al. \cite{LNPSY21} to obtain a minor of $B$, which we call a \textit{kernel}, which has the following two properties with high probability:
\begin{enumerate}
    \item A min-weight separator in the kernel can be used to find a min-weight separator in $B$ in $\tilde{O}(p)$ time.
    \item The kernel has at most $9rs^r$ vertices from $U_V$.
\end{enumerate}
The first property helps in reducing the problem of finding a min-weight separator in $B$ to the problem of finding a kernel and of finding a min-weight separator in the kernel. The second property allows us to design an algorithm which finds a min-weight separator in a kernel in time $\tilde{O}(\lambda r^2s^{3r})$.
In Algorithm \ref{alg:find-kernel}, we give a randomized procedure that takes $B$, a set $X \subseteq V$ such that $X \cap L \neq \emptyset$, and an integer $\ell$ with $|L|/2 \leq \ell \leq |L|$ as inputs and returns a kernel. Running this algorithm $O(\log n)$ times with $X = U_V$ for each $\ell \in \{1, 2, 4, \dots, 2^{\lceil\log(3s^r)\rceil} \}$ will return a kernel. We show the correctness of Algorithm \ref{alg:find-kernel} in Lemma \ref{lem:alg-finds-kernel}. Unfortunately, the runtime of Algorithm \ref{alg:find-kernel} is $\tilde{O}(|E_H| \cdot |X|)$, which is too slow for our purposes. In Lemma \ref{lem:efficient-kernel}, we show that there is a faster algorithm that returns the same kernel as Algorithm \ref{alg:find-kernel}. Finally, in Lemma \ref{lem:alg-for-min-separator} we
obtain an $\tilde{O}(s^{3r}p)$ time algorithm for finding a min-weight separator in $B$ by combining the algorithm for finding a kernel with an algorithm due to Dinitz \cite{Dinitz06} which finds a min-weight separator in a kernel.



We will need the following definitions (we recall that the graph $B$ has $m+n$ vertices):

\begin{definition} [\cite{LNPSY21}]
\begin{enumerate}
    \item For $X, Y\subseteq V_B$, a separator $(L,S,R)$ is an \emph{$(X,Y)$-separator} if $X \subseteq L, Y \subseteq R$. 
    \item A minimum weight $(X,Y)$-separator will be denoted as a \emph{min $(X,Y)$-separator}. 
    \item For a positive integer $t$, a separator $(L,S,R)$ is a \emph{$t$-scratch} if $|S| \leq t$, $|L| \leq t/(100 \log (m+n))$, and $|S_{low}| \geq 300|L|\log (m+n)$, where $S_{low} := \{v \in S : \deg(v) \leq 8t\}$. 
\end{enumerate}

\end{definition}

\begin{definition}\label{def:kernel}
Let $x \in V_B$. 
We say that $(B_x,Z_x,x,t_x)$ is a kernel if the following hold: 
\begin{enumerate}
    \item $B_x$ is a subgraph of the graph obtained from $B$ by contracting some subset $T_x \subseteq V_B\setminus \{x\}$ 
    into the single vertex $t_x$ such that $B_x$ contains $x$ and $t_x$.
    \item $Z_{x}$ is a subset of $V_B$ such that $(L,S,R)$ is a min $(\{x\},\{t_x\})$-separator in $B_{x}$ iff $(L',S \cup Z_{x},R')$ is a min $(\{x\},T_x)$-separator in $B$, where $L'$ is the component of $V_B \setminus (S \cup Z_{x})$ containing $x$, and $R' := V_B \setminus (L' \cup S \cup Z_x)$. 
\end{enumerate}
\end{definition}

Let $(L,S,R)$ be a min-weight separator in $B$, and let $x \in L$. The following lemma gives us a way to choose a set $T_x$ so that with constant probability $\emptyset \neq T_x \subsetneq R$. For such a $T_x$, we have that $(L,S,R)$ is a min $(\{x\}, T_x)$-separator, which means that we can potentially recover $(L,S,R)$ from the min-weight separator in a kernel.
\begin{lemma}\label{lem:Tx-subset-R}
Let $(L,S,R)$ be a min-weight separator in $B$ with $|L|\le 3s^r$. Let $x \in L$ and let $\ell$ be a positive integer with $|L|/2 \leq \ell \leq |L|$. Let $T$ be a subset of $V(B)$ chosen by sampling each vertex with probability $1/(8\ell)$, and let $T_x = T \setminus N[x]$. Then, we have that $\emptyset \neq T_x \subsetneq R$ with probability $\Omega(1)$.
\end{lemma}
\begin{proof}
By the hypothesis of Theorem \ref{thm:big-lambda-small-cut}, we have that $|S| = \lambda \geq 900s^r\log n$. 
We also have that 
$|L| \leq 3s^r \leq |S|/2$. We also note that a random separator consisting of the neighbors of a single vertex of $U_V$ has expected size at most $r|U_E|/|U_V|$, so $|S| \leq (r/n)|U_E| \leq |U_E|/2$. Consequently, we have that 
\[|(U_V \cup U_E) \setminus N[x]| \geq |U_E| - |S| - |L|  \geq |U_E| - (|U_E|/2) - (|U_E|/4) \geq |U_E|/4 \geq |S|/4 \geq s^r.\] Since each vertex of $V(B) \setminus N[x]$ is included in $T_x$ with probability $1/(8\ell) \geq 1/24s^r$, it follows that at least one of them will be included in $T_x$ with at least constant probability.

Now we note that since $N(x)$ is a separator in $B$, $|N(x)| \geq \lambda$. Thus, $|(L \cup S) \setminus N[x]| \leq (|L| + \lambda) - \lambda = |L|$. Since each of these vertices is included in $T_x$ with probability at most $1/(8(|L|/2)) = 1/(4|L|)$, the probability that none of them are included is at least a constant.

Since $|R| \geq 1$ by the definition of a separator, and since each vertex of $R$ is included in $T_x$ with probability $1/(8\ell) \leq 1/8$, the probability that at least one vertex of $R$ is excluded from $T_x$ is also at least a constant. 
\end{proof}

Our technique to find a min-weight separator in $B$ is to 
find, for each $x \in U_V$, a kernel with $O(\lambda r s^{2r})$ edges. Algorithm \ref{alg:find-kernel} is a slightly modified version of an algorithm of Li et al. \cite{LNPSY21} for computing a small kernel. The difference is that while our algorithm contracts $T_x' := T_x \cup (N(T_x) \cap U_V)$, their algorithm merely contracts $T_x$.
We need this additional contraction in order to prove Lemma \ref{lem:kernel-is-small}, which bounds the size of the kernel. This bound on the size of the kernel that we obtain is necessary to prove Claim \ref{clm:min-separator-in-kernel}, which shows that we can find a min-weight separator in a kernel efficiently.
Additionally, while Li et al. analyzed the algorithm only for an unweighted graph, we run it on the graph $B$, which has vertex weights (of a special kind). 

\begin{algorithm}[ht]
\caption{{FindKernel}($B,X,\ell$)}\label{alg:find-kernel}
\begin{algorithmic}
\State $T \gets \emptyset$
\For{$v$ in $V(B)$}
\State Add $v$ to the set $T$ with probability $1/(8\ell)$
\EndFor
\For{$x \in X$}
\State $T_x \gets T  \setminus N[x]$
\State $T_x' \gets T_x \cup (N(T_x) \cap U_V)$
\State Obtain $B'$ from $B$ by contracting the vertices of
$T_x'$
into a single vertex $t_x$
\State $Z_{x} \gets \{v \in V(B') \colon \{x,v\},\{v,t_x\} \in E(B')\}$
\State Obtain $B''$ from $B'$ by deleting the vertices of $Z_{x}$ and the edges that are fully contained in $N(t_x)$ 
\State Obtain $B'''$ from $B''$ by deleting all vertices which are disconnected from $x$ in $B'' \setminus N[t_x]$
\State Obtain $B_{x}$ from $B'''$ by deleting from $B'''$ all neighbors of $t_x$ with degree one 
\EndFor
\State \Return $\{ (B_{x},Z_{x},x,t_x) \colon x \in X\}$
\end{algorithmic}
\end{algorithm}

Now we show the correctness of Algorithm \ref{alg:find-kernel}.

\begin{lemma}\label{lem:alg-finds-kernel}
Consider executing Algorithm \ref{alg:find-kernel} on input $(B, X, \ell)$ for some positive integer $\ell$ such that $|L|/2 \le \ell \le |L|$, where $X$ is a subset of $U_V$ and $(L, S, R)$ is a min-weight separator in $B$. For $x\in X\cap L$, if the set $T_x$ chosen by the algorithm satisfies $\emptyset \neq T_x \subsetneq R$, then the tuple $(B_x, Z_x,x,t_x)$ returned by the algorithm is a kernel.

\end{lemma}
\begin{proof}
Suppose that $\emptyset \neq T_x \subsetneq R$. Then, $(L,S,R)$ is a min $(\{x\},T_x)$-separator, and contracting $T_x$ cannot destroy $(L,S,R)$ (or any other min $(\{x\},T_x)$-separator). Let $R'$ be $R$ with the vertices of $T_x'=T_x \cup (N(T_x) \cap U_V)$ contracted into the single vertex $t_x$. By Observation \ref{obs:subset-of-ue} we have that $S \subseteq U_E$, which means that since $T_x \subseteq R$ then $N(T_x) \cap U_V \subseteq R$ as well. Thus, $(L,S,R')$ is a min-weight separator in $B'$. Therefore, $(B', \emptyset, x, t_x)$ is a kernel. 

By Observation \ref{obs:subset-of-ue}, we have that the weight of the separator $(L,S,R)$ is $|S|$. Consequently, we have the following observation:

\begin{observation}\label{obs:disjoint-paths}
There exist $|S|$ vertex-disjoint $x-t_x$ paths in $B'$, each of which visits exactly one vertex from $N(x)$ and exactly one vertex from $N(t_x)$.
\end{observation}

Since every edge in $B'[N(x)]$ or $B'[N(t_x)]$ is not present in any of these paths, removing all of these edges does not affect the weight of a min $(\{x\},\{t_x\})$-separator in $B'$. On the other hand, every vertex $v \in N(x) \cap N(t_x)$ is included in every $(\{x\},\{t_x\})$-separator in $B'$. Therefore, $(B'',Z_{x},x,t_x)$ is a kernel. 

By Observation \ref{obs:disjoint-paths}, removing from $B''$ vertices which are disconnected from $x$ in $B'' \setminus N[t_x]$ cannot affect the size of a min $(\{x\},\{t_x\})$-separator, since such vertices cannot participate in any $s-t_x$ paths which use only a single vertex of $N(t_x)$. Thus, $(B''',Z_{x},x,t_x)$ is also a kernel.

Finally, we note that degree one neighbors of $t_x$ cannot participate in any $x-t_x$ paths using only one vertex from $B'''[N[t_x]]$. Thus, again by Observation \ref{obs:disjoint-paths}, removing such vertices does not affect the weight of a min $(\{x\},\{t_x\})$-separator, and thus $(B_{x},Z_{x},x,t_x)$ is a kernel.
\end{proof}

By Lemmas \ref{lem:Tx-subset-R} and \ref{lem:alg-finds-kernel}, running Algorithm \ref{alg:find-kernel} $O(\log n)$ times for each $\ell\in \{1, 2, 4, 8, ..., 2^{\lceil\log(3s^r)\rceil}\}$ and $X=U_V$ will return a kernel $(B_x, Z_x,x,t_x)$ for some $x\in L$ with high probability, where $(L, S, R)$ is a min-weight separator of $B$ with $|L|\le 3s^r$. 
We next show that the kernel computed by Algorithm \ref{alg:find-kernel} is not too large.
\begin{lemma}\label{lem:kernel-is-small}
Consider executing Algorithm \ref{alg:find-kernel} on input $(B, X, \ell)$ for some positive integer $\ell$ such that $|L|/2 \le \ell \le |L|$, where $X$ is a subset of $U_V$ and $(L, S, R)$ is a min-weight separator in $B$ with $|L|\le 3s^r$. For $x\in X\cap L$, if the set $T_x$ chosen by the algorithm satisfies $\emptyset \neq T_x \subsetneq R$ and the tuple $(B_x, Z_x,x,t_x)$ returned by the algorithm is a kernel, then $B_{x}$ contains at most $9rs^r$ vertices from $U_V$ with high probability.
\end{lemma}

\begin{proof}
Consider $(L', S', R')$ where $S':= S \setminus Z_x$, $L'$ is the component of $V[B_x] \setminus S'$ which contains $x$, and $R' := V[B_x] \setminus (L' \cup S')$. 
Since $T_x\subseteq R$ and $(B_x, Z_x,x,t_x)$ is a kernel, we have that $(L',S',R')$ is a min-weight $(x,t_x)$ separator. We know that $|L' \cap U_V| \leq |L| \leq 3s^r \leq rs^r$ and that $S' \subseteq U_E$, so it suffices to show that $|R'\cap U_V| \leq 8rs^r$.


Consider an arbitrary vertex $v \in R' \cap U_V$. Since $v$ was not contracted into $t_x$, it must be the case that $v$ was not in $T_x$, and none of $v$'s neighbors were in $T_x$. If $v$ had more than $8\ell \log n$ neighbors in $V(B) \setminus N(x)$, then with high probability, one of the neighbors of $v$ would be in $T_x$. Therefore, with high probability, $v$ has at most $8\ell \log n$ neighbors in $V(B) \setminus N(x)$. We know, however, that $v$ has at least $\lambda$ neighbors in $B$. Therefore, $v$ has at least $\lambda - 8\ell \log n$ neighbors in $N(x)$.

Since $x \in L$, we have that $|N(x)| \leq |L| + \lambda \leq 3s^r + \lambda$. Since $x \in U_V$, we also have that $N(x) \subseteq U_E$. Since $G$ has rank $r$, it follows that each vertex in $N(x)$ has degree at most $r$. Thus, the vertices of $N(x)$ collectively have at most $r(3s^r +\lambda)$ edges incident to them. Since each vertex of $R' \cap U_V$ is incident to at least $\lambda - 8\ell \log n$ of those edges, we have that 
\[|R' \cap U_V| \leq \frac{r(3s^r+\lambda)}{\lambda - 8\ell \log n}.\]

From the hypothesis of Theorem \ref{thm:big-lambda-small-cut}, we have that $\lambda \geq 900s^r \log n$. From the hypothesis of the claim we are proving, we have that $\ell \leq 3s^r$. Therefore, $\lambda - 8\ell \log n \geq \lambda /2$, so

\[
\frac{r(3s^r+\lambda)}{\lambda - 8\ell \log n} \leq \frac{2r(3s^r+\lambda)}{\lambda} \leq 2r(3s^r+1) \leq 8rs^r.\qedhere
\]
\end{proof}

The following lemma will be useful in showing that the kernel generated by Algorithm \ref{alg:find-kernel} can be obtained efficiently.
\begin{lemma}\label{lem:efficient-kernel}
There is a randomized algorithm that takes as input the bipartite graph $B$, a subset $X\subseteq U_V$, and positive integers $\ell, t$, 
chooses a subset $T$ of $V(B)$ by sampling each vertex with probability $1/(8\ell)$, 
and runs in time $\tilde{O}(|E(H)| + |X|t\ell)$ to return a tuple $(B_x, Z_x,x,t_x)$ for every $x\in X$. Additionally, the returned collection satisfies the following with high probability: 
\begin{itemize}
    \item For $x\in X$, if $B$ contains a $t$-scratch $(L,S,R)$ with $x \in L$, $|L|/2 \leq \ell \leq |L|$, and $\emptyset \neq T \setminus N[x] \subseteq R$, then the tuple $(B_{x},Z_{x},x,t_x)$ returned by the algorithm is a kernel with $|E(B_{x})| = O(t\ell \log p)$ and $B_x$ contains at most $9rs^r$ vertices from $U_V$.
\end{itemize}


\end{lemma}
\begin{proof}
This follows directly from the results of \cite{LNPSY21}. They design an algorithm (Algorithm 1 in \cite{LNPSY21}) which chooses a set $T$ according to the same distribution as Algorithm \ref{alg:find-kernel}, but uses linear sketching and BFS-like computations starting from $x$ to compute a tuple $(B_{x},Z_{x},x,t_x)$ more efficiently than Algorithm \ref{alg:find-kernel}; moreover, if $B$ contains a $t$-scratch $(L,S,R)$ with $x \in L, |L|/2 \leq \ell \leq |L|$, and $\emptyset\neq T \setminus N[x] \subseteq R$, then with high probability the tuple $(B_x,Z_x,x,t_x)$ returned by their algorithm is a kernel with $|E(B_{x})| = O(t\ell \log p)$, and this kernel is the tuple that would be output by their version of Algorithm \ref{alg:find-kernel} when run on the same input with the same source of randomness.


We obtain an algorithm with the guarantees given in the statement of the lemma by applying the modifications described in \cite{LNPSY21} to our version of Algorithm \ref{alg:find-kernel}. There are only two differences between Algorithm \ref{alg:find-kernel} and the corresponding algorithm in \cite{LNPSY21}. First, their algorithm is only for unweighted graphs, while we run Algorithm \ref{alg:find-kernel} on the weighted graph $B$. Second, we contract $T_x' := T_x \cup (N(T_x) \cap U_V)$ rather than just $T_x$. We now show that neither of those differences affects the correctness or runtime of the modified algorithm.

We note that since Algorithm \ref{alg:find-kernel} does not even look at the weights on the vertices, the fact that $B$ is weighted does not affect the runtime of our algorithm. 

The additional contraction of $N(T_x) \cap U_V$ into $t_x$ can be performed without increasing the runtime of the algorithm, and as shown in Lemma \ref{lem:alg-finds-kernel} the resulting algorithm still finds a kernel whenever $\emptyset \neq T_x \subsetneq R$. Furthermore, by Lemma \ref{lem:kernel-is-small}, this kernel contains at most $9rs^r$ vertices from $U_V$ with high probability. We skip the detailed description and analysis of the faster algorithm for finding kernels, since it is identical to \cite{LNPSY21}.
\end{proof}

The following lemma completes the proof of Theorem \ref{thm:big-lambda-small-cut}.
\begin{lemma}\label{lem:alg-for-min-separator}
There exists a randomized algorithm that runs in time $\tilde{O}(s^{3r}p)$ to find a min-weight separator in $B$ with high probability.
\end{lemma}

\begin{proof}
Let $k$ be the approximation of $\lambda$ obtained by running the algorithm from Theorem \ref{thm:CX-approximation}. To find a min-weight separator in $B$, we run the algorithm from Lemma \ref{lem:efficient-kernel} on $B$ with $t = k+r+300s^r\log n$, and $X = U_V$ for $\Theta(\log n)$ times for each $\ell \in \{ 1,2, 4,8, \dots, 2^{\lceil \log(3s^r)\rceil}\}$. Then, for each tuple $(B_{x}, Z_{x},x,t_x)$ returned, if $B_{x}$ contains at most $9rs^r$ vertices from $U_V$ then we compute a min $(\{x\}, \{t_x\})$-separator in $B_{x}$. For each such separator $(L,S,R)$, we compute a corresponding separator $(L', S', R')$ in $B$ by setting $S' := S \cup Z_x$, $L'$ to be the component of $V_B \setminus S'$ which contains $x$, and $R' := V_B \setminus (L' \cup S')$. By Definition \ref{def:kernel}, if $(B_x, Z_x,x,t_x)$ is a kernel then $(L',S',R')$ is a min-weight $(\{x\},T_x)$-separator in $B$.
We return the cheapest among all computed separators.

We first analyze the correctness of our procedure. We will need the following claim:

\begin{claim}\label{clm:min-cut-is-scratch}
Let $(L,S,R)$ be a min-weight separator in $B$ with $|L| \leq 3s^r$. Then, $(L,S,R)$ is a $t$-scratch.
\end{claim}

\begin{proof}
Since each min-weight separator in $B$ corresponds to a min-cut in $G$, we know that $|S| = \lambda \leq k \leq t$. By assumption and by our choice of $t$, we have that $|L| \leq 3s^r \leq t/(100\log n)$. By Observation \ref{obs:subset-of-ue}, we have that $S \subseteq U_E$. Also, by the definition of $B$, every vertex in $U_E$ has degree at most $r$, which is less than $8t$. Therefore, we have that $S_{low} = S$. Thus, $|S_{low}| = |S| = \lambda \geq 900s^r\log n \geq 300|L| \log n$, so $(L,S,R)$ is a $t$-scratch.\end{proof}

Let $(L, S, R)$ be a min-weight separator in $B$ with $|L|\le 3s^r$. By Claim \ref{clm:min-cut-is-scratch}, $(L,S,R)$ is a $t$-scratch. Since we run the algorithm of Lemma \ref{lem:efficient-kernel} $\Theta(\log n)$ times 
for each $\ell = 1,2,4,8, \dots, 2^{\lceil \log(3s^r)\rceil}$, we will run the algorithm $\Theta(\log n)$ times for some 
$\ell \in [|L|/2, |L|]$. By Lemma \ref{lem:Tx-subset-R}, for each $x \in L$, each time we run the algorithm, the set $T$ that is chosen satisfies $\emptyset \neq T_x \subsetneq R$ with constant probability, so with high probability this occurs for some $x \in L$ at least once. Lemma \ref{lem:efficient-kernel} tells us that when the algorithm chooses such a $T$, then with high probability it will return a kernel $(B_x, Z_x,x,t_x)$ where $B_x$ has at most $9rs^r$ vertices from $U_V$. 
By Definition \ref{def:kernel}, a min $(\{x\},\{t_x\})$-separator $(L, S, R)$ in this kernel gives a min $(\{x\}, T_x)$-separator $(L', S\cup Z_x, R')$ in $B$, where $L'$ is the component of $V_B\setminus (S\cup Z_x)$ containing $x$, and $R'=V_B\setminus (L'\cup S\cup Z_x)$. Since $T_x \subseteq R$, a min $(\{x\},T_x)$-separator in $B$ is in fact a min-weight separator in $B$. Thus, the algorithm finds a min-weight separator in $B$ with high probability.

We now analyze the runtime of our procedure. The algorithm from Lemma \ref{lem:efficient-kernel} runs in time $\tilde{O}(|E_B| + |U_V|(k+r+300s^r)(3s^r)) = \tilde{O}(p + n(\lambda+r+s^r)s^r) = \tilde{O}(p+ps^r +rns^{2r}) = \tilde{O}(ps^r+rns^{2r})$. Running the algorithm $\Theta(\log n)$ times each for $\log(3s^r)$ different values of $\ell$ increases the running time by a factor of $O(\log(3s^r)\log n) = O(r(\log s)(\log n)) = O(r\log^2 n)$ giving us a total runtime of $\tilde{O}(prs^r + r^2ns^{2r})$. To bound the time to compute a min $(\{x\},\{t_x\})$-separator in each kernel we use the following claim.
\begin{claim}\label{clm:min-separator-in-kernel}
If $(B_{x},Z_{x},x,t_x)$ is a tuple returned by the algorithm from Lemma \ref{lem:efficient-kernel} with at most $9rs^r$ vertices from $U_V$ and $O(\lambda r s^{2r})$ edges, then a min $(\{x\},\{t_x\})$-separator in $B_{x}$ can be computed in time $\tilde{O}(\lambda r^2s^{3r})$.
\end{claim}

\begin{proof}
Since $B$ is bipartite, every path of length $d$ in $B$ visits at least $\lceil d /2 \rceil$ vertices in $U_V$. The graph $B_{x}$ is obtained from a bipartite graph by contracting some subset of vertices. Therefore, except for the vertex $t_x$ created by the contraction, every vertex in $B_{x}$ which is not in $U_V$ must have all of its neighbors in $U_V$. This means that for every path in $d$, every pair of consecutive vertices in the path not including the vertex generated by the contraction has at least one vertex from $U_V$. Therefore, every path of length $d$ in $B_{x}$ visits at least $\lceil (d-2)/2 \rceil$ vertices from $U_V$. By the assumption of the claim, $B_{x}$ contains $O(rs^r)$ vertices from $U_V$. Thus, every path in $B_{x}$ has length $O(rs^r)$. 

Therefore, we can use Dinitz's algorithm to efficiently find a min $(\{x\},\{t_x\})$-separator in $B_{x}$ \cite{Dinitz06}. Since the lengths of paths in $B_{x}$ are bounded by $O(rs^r)$, the number of blocking flows found by Dinitz's algorithm is also $O(rs^r)$. Each blocking flow can be found in time proportional to the number of edges in $B_{x}$, and by 
assumption this is $\tilde{O}(\lambda rs^{2r})$.
Therefore, the total runtime of Dinitz's algorithm on a $B_x$ is $\tilde{O}(\lambda r^2s^{3r})$. 
\end{proof}

Let $(B_x, Z_x,x,t_x)$ be a kernel returned by the algorithm from Lemma \ref{lem:efficient-kernel}. By Lemma \ref{lem:efficient-kernel}, we have that $B_x$ contains at most $9rs^r$ vertices from $U_V$ and $O(\lambda r s^{sr})$ edges with high probability. By Claim \ref{clm:min-separator-in-kernel}, we can compute a min-weight separator for a single $B_x$ with at most $9rs^r$ vertices from $U_V$ and $O(\lambda r s^{sr})$ edges in time $\tilde{O}(\lambda r^2s^{3r})$. Thus, we can compute a min-weight separator for all $B_x$ with at most $9rs^r$ vertices from $U_V$ and $O(\lambda r s^{sr})$ edges in time $\tilde{O}_r(s^{3r}\lambda n) = \tilde{O}_r(s^{3r}p)$. 

\end{proof}

\subsection{Expander Decomposition Based Min-Cut}\label{sec:large-min-cut}
In this section, we give an algorithm to find connectivity in hypergraphs in which every min-cut has large size. In order to analyze the run-time of our algorithm, we need to bound the run-time to implement trim and shave operations. This is summarized in the following claim.

\begin{claim}\label{clm:efficient-trim-and-shave}
Let $G=(V,E)$ be an $n$-vertex $r$-rank hypergraph of size $p$ with $p \geq n$, and let $\mathcal{X}=\{X_1, \dots, X_k\}$ be a collection of disjoint subsets of $V$. Then, there exists a deterministic algorithm which takes as input $G$  and $\mathcal{X}$, and runs in time $O(pr)$ to return $\{\trim(X_1), \dots, \trim(X_k)\}$ and $\{\shave(X_1), \dots, \shave(X_k)\}$.
\end{claim}
\begin{proof}

For $v\in V$, let $part(v):=i$ if $v\in X_i$ and $0$ otherwise. For $e\in E$, let $epart(e):=i$ if $e\subseteq X_i$ and $0$ otherwise. Let $d_{\mathcal{X}}(v):=d_{X_i}(v)$ if $v\in X_i$ and $0$ otherwise and $\delta_{\mathcal{X}}(v):=\delta_{X_i}(v)$ if $v\in X_i$ and $\emptyset$ otherwise. We also recall that the degree of a vertex in $G$ is denoted by $d(v)$. The functions $part:V\rightarrow \{0, \dots, k\}$, $d: V \rightarrow \mathbb{Z}_+$, $d_{\mathcal{X}} : V \rightarrow \mathbb{Z}_{\geq 0}$, and $\delta_{\mathcal{X}}: V \rightarrow 2^E$ can be computed in $O(p)$ time using Algorithm \ref{alg:trim-shave-helper}. 

\begin{algorithm}\label{alg:trim-shave-helper}
\caption{TrimShaveHelper$(G=(V,E),\mathcal{X}=\{X_1, \dots, X_k\})$}
\begin{algorithmic}
\ForAll{$v \in V$}
\State $part(v), d(v), d_{\mathcal{X}}(v) \gets 0$
\State $\delta_{\mathcal{X}}(v) \gets \emptyset$
\EndFor
\For{$i=1,\ldots, k$}
\ForAll{$v \in X_i$}
\State $part(v) \gets i$
\EndFor
\EndFor
\ForAll{$e \in E$}
\State $v \gets$ an arbitrary vertex of $e$
\State $epart(e) \gets part(v)$
\ForAll{$u \in e \setminus \{v\}$}
\If{$part(u) \neq part(v)$}
\State $epart(e) \gets 0$
\EndIf
\EndFor
\If{$epart(e) \neq 0$}
\ForAll{$v \in e$}
\State $d_{\mathcal{X}}(v) \gets d_{\mathcal{X}}(v) + 1$
\State Add $e$ to $\delta_{\mathcal{X}}(v)$
\EndFor
\EndIf
\EndFor
\State \Return $part,d,d_X,\delta_X$
\end{algorithmic}
\end{algorithm}

We claim that Algorithm \ref{alg:trim-shave-helper} can be implemented to run in time $O(p)$. We note that looping over all of the vertices takes time $O(n)$. Since the sets $X_i$s are disjoint, looping over all the vertices of each $X_i$ also takes time $O(n)$. Looping over all edges and within each edge looping over all vertices in the edge takes time $O(\sum_{e \in E} |e|) = O(p)$. Since $n \leq p$, the complete algorithm can be implemented to run in time $O(p)$.

We use Algorithm \ref{alg:shave} to perform the $\shave$ operation. Since the sets $X_i$s are disjoint, iterating over all vertices of each $X_i$ takes time $O(n)$, so the portion of the algorithm outside of the call to \textsc{TrimShaveHelper} can be implemented to run in time $O(p)$, and therefore the whole algorithm can be as well.

\begin{algorithm}
\caption{ShaveAlgorithm$(G=(V,E), \mathcal{X}=\{X_1, \dots, X_k\}$}\label{alg:shave}
\begin{algorithmic}
\State $part,d,d_{\mathcal{X}},\delta_{\mathcal{X}} \gets \textsc{TrimShaveHelper}(G,\mathcal{X})$
\For{$i=1, \dots, k$}
\ForAll{$v \in X_i$}
\If{$d_{\mathcal{X}}(v) \leq (1-1/r^2)d(v)$}
\State Remove $v$ from $X_i$
\EndIf
\EndFor
\EndFor
\State \Return $\mathcal{X}$
\end{algorithmic}
\end{algorithm}

\begin{algorithm}
\caption{TrimAlgorithm$(G=(V,E), \mathcal{X}=\{X_1, \dots, X_k\})$} \label{alg:trim}
\begin{algorithmic}
\State $x,y,d,d_{\mathcal{X}},\delta_{\mathcal{X}} \gets \textsc{TrimShaveHelper}(G,\mathcal{X})$
\State $T \gets \emptyset$
\State $F \gets V$
\While{$T \cup F \neq \emptyset$}
\If{$T \neq \emptyset$}
\State Let $v \in T$ be arbitrary
\State Remove $v$ from $T$ and from $X_{part(v)}$
\State $part(v) \gets 0$
\ForAll{$e \in \delta_{\mathcal{X}}(v)$}
\For{$u \in e \setminus \{v\}$}
\State Remove $e$ from $\delta_{\mathcal{X}}(u)$
\State $d_{\mathcal{X}}(u) \gets d_{\mathcal{X}}(u) - 1$
\If{$d_{\mathcal{X}}(u) < d(u)/2r$ and $u \not\in T$}
\State Add $u$ to $T$
\EndIf
\EndFor
\EndFor
\Else
\State Let $v \in F$ be arbitrary
\State remove $v$ from $F$
\If{$d_{\mathcal{X}}(v) < d(v)/2r$ and $part(v) \neq 0$}
\State Add $v$ to $T$
\EndIf
\EndIf
\EndWhile
\State \Return $\mathcal{X}$
\end{algorithmic}
\end{algorithm}
We use Algorithm \ref{alg:trim} to perform the $\trim$ operation. Here the set $F$ stores vertices that the algorithm has yet to examine, and $T$ stores vertices that the algorithm has determined must be trimmed. The values of $d_{\mathcal{X}}(v)$ and $\delta_{\mathcal{X}}(v)$ are updated as vertices are trimmed, and whenever $d_{\mathcal{X}}(v)$ changes for a vertex $v$, the algorithm verifies whether that vertex $v$ should be trimmed. Removing vertices from $X_i$ cannot increase $d_{X_i}(v)$ for any vertex $v \in X_i$. Therefore, if a vertex $v$ satisfies $d_{X_i}(v) < d(v) / 2r$ when $v$ is added to $T$, $v$ will still satisfy this inequality when it is actually removed from $X_i$.
We note that each vertex is chosen at most once as the arbitrary vertex from $F$, since whenever a vertex is chosen from $F$ it is removed from $F$ and never added back to it. Whenever a vertex $v$ is chosen from $T$, it is removed from $T$, and $part(v)$ is set to $0$. Since the algorithm only adds a vertex $v$ to $T$ only if $part(v) \neq 0$, each vertex is chosen at most once as the arbitrary vertex from $T$.

The case of choosing a vertex from $F$ can be implemented to run in constant time since we perform an arithmetic comparison and a constant number of set additions and removals. 
When the algorithm selects a vertex from $T$, it iterates over all hyperedges of that vertex which are currently in $X_{part(v)}$. For each of these hyperedges $e$, all iterations of the innermost \textbf{for} loop can be implemented to run in time $O(|e|)$. Thus, the total time spent processing a vertex $v$ from $T$ is $O(\sum_{e \in \delta(v)} |e|)$. Since each vertex is processed in $T$ at most once, the total time to process all of the vertices is $O(\sum_{v \in V}\sum_{e\in \delta(v)}|e|)=O(pr)$.
\end{proof}


The following theorem is the main result of this section. 
\begin{theorem}\label{thm:large-mincut-algo}
Let $G=(V,E)$ be an $n$-vertex $r$-rank simple hypergraph of size $p$ with connectivity $\lambda$ such that 
$\lambda \ge r(4r^2)^r$ 
and every min-cut $(C,V \setminus C)$ in $G$ has $\min\{|C|,|V \setminus C|\} > r-\log (\lambda/4r) / \log n$. Then, there exists a randomized algorithm which takes $G$ as input
and runs in time 
\[
\tilde{O}_r\left(p + (\lambda n)^{1+o(1)} + 
\min\left\{\lambda^{\frac{r-3}{r-1}+o(1)}n^{2+o(1)},
\frac{n^{r+r \cdot o(1)}}{\lambda^{\frac{r}{r-1}}},  \lambda^{\frac{5r-7}{4r-4}+o(1)}n^{\frac{7}{4}+o(1)}\right\}\right),
\]
to return a min-cut of $G$ with high probability.

\end{theorem}

\begin{proof}
We will use algorithms from Theorems \ref{thm:CX-approximation}, \ref{thm:k-certificate}, \ref{thm:FPZ-min-cut}, and \ref{thm:CQ-min-cut}, Corollary \ref{thm:CX-min-cut}, and Lemma \ref{lem:efficient-expander-decomp} to find a min-cut. We use \textsc{CXApproximation} to find a constant-approximation $k$ for the min-cut capacity. Next, we use \textsc{Certificate}$(G,k)$ to get a graph $G'$ with size $p_1=O(rkn)=O(r\lambda n)$. Next, we find an expander decomposition $\cX$ of $G'$ with respect to the parameter $\phi:=(6r^2/\delta)^{1/(r-1)}$. We then apply the trim operation on $\cX$ followed by $3r^2$ shave operations to obtain a collection $\cX''$ of disjoint subsets of $V$. We contract all these subsets and run either \textsc{FPZMinCut} or \textsc{CXMinCut} or \textsc{CQMincut} on the resulting hypergraph to find a min-cut. We present the pseudocode of our approach in Algorithm \ref{alg:exp-decomp-min-cut}. 

\begin{algorithm}
\caption{ExpDecompMinCut$(G)$}\label{alg:exp-decomp-min-cut}
\begin{algorithmic}
\State $k \gets \textsc{CXApproximation}(G)$
\State $G' \gets \textsc{Certificate}(G,k)$
\State $\delta' \gets \min_{v \in V} d_{G'}(v)$
\State $\phi \gets \min\{(6r^2/\delta')^{1/(r-1)}, 1/(r-1)\}$
\State $\mathcal{X} \gets$ \textsc{ExpanderDecomposition}$(G',\phi)$
\State $\mathcal{X'} \gets \textsc{TrimAlgorithm}(G',\mathcal{X})$
\State $\mathcal{X}'' \gets \mathcal{X'}$
\For{$i=1, \ldots, 3r^2$}
\State $\mathcal{X}'' \gets \textsc{ShaveAlgorithm}(\mathcal{X''})$
\EndFor
\State Let $G''$ be the hypergraph obtained by contracting each set $X'' \in \mathcal{X''}$ into a single vertex
\State Run $\text{CXMinCut}(G'')$, $\text{FPZMinCut}(G'')$, and $\text{CQMinCut}(G'')$ simultaneously, return the output of the earliest terminating among the three, and terminate
\end{algorithmic}
\end{algorithm}


We now prove the correctness of Algorithm \ref{alg:exp-decomp-min-cut}. 
By hypothesis, we know that every min-cut $(C, V\setminus C)$ has $\min\{|C|, |V\setminus C|\}>r-\log{(\lambda/4r)}/\log{n}$. Hence, by Lemma \ref{cor:small-side-not-tiny-implies-small-side-large}, every min-cut $(C, V\setminus C)$ has $\min\{|C|, |V\setminus C|\}>(\lambda/2)^{1/r}\ge 4r^2$ (by assumption on $\lambda$). 

By Theorem \ref{thm:CX-approximation}, we have that $k > \lambda$. Therefore, by Theorem \ref{thm:k-certificate}, the subhypergraph $G'$ is a simple hypergraph with connectivity $\lambda$ and the min-cuts in $G'$ are exactly the min-cuts of $G$. Consequently, the min-degree $\delta'$ in $G'$ is at least the connectivity $\lambda$. Let $(C, V \setminus C)$ be an arbitrary min-cut in $G'$ and let $\mathcal{X}=(X_1, \ldots, X_t)$ be the expander decomposition of $G'$ for the parameter $\phi$. By definition, $\phi\le 1/(r-1)$. Hence, by Lemma \ref{lem:efficient-expander-decomp}, for every $i \in [t]$, we have that 
\begin{align*}
\delta' \geq \lambda &\geq \left|E^o_{G'}(X_i \cap C, X_i \cap (V \setminus C))\right| \\
&\geq \min\left\{\left(\frac{6r^2}{\delta'}\right)^{\frac{1}{r-1}},\frac{1}{r-1}\right\} \min \{\vol_{G'}(X_i \cap C), \vol_{G'}(X_i \setminus C)\} \\
&\geq \min\left\{\left(\frac{6r^2}{\delta'}\right)^{\frac{1}{r-1}},\frac{1}{r-1}\right\}  \delta' \min \{|X_i \cap C|, |X_i \setminus C|\}.
\end{align*}
Thus, $\min\{|X_i \cap C|, |X_i \setminus C| \} \leq \max\{(\delta'/6r^2)^{1/(r-1)},r-1\}$ for every $i\in [t]$. 
If $\max\{(\delta'/6r^2)^{1/(r-1)},r-1\}=r-1$, then $\min\{|X_i' \cap C|, |X_i' \setminus C| \} \le r-1\le 3r^2$ for every $i\in [t]$ and if $\max\{(\delta'/6r^2)^{1/(r-1)},r-1\}=(\delta'/6r^2)^{1/(r-1)}$, then by Claim \ref{clm:trim-bound}, we have that $\min\{|X_i' \cap C|, |X_i' \setminus C| \} \le 3r^2$ for every $i\in [t]$. 
We recall that $\lambda \ge r(4r^2)^r$ and every min-cut has size at least $4r^2$. Hence, by 
$3r^2$ repeated applications of Claim \ref{clm:shave-bound}, we have that $\min\{|X''_i \cap C|, |X''_i \cap (V \setminus C)|\} = 0$ for every $i\in [t]$. 
In particular, this means that contracting each $X''_i$ does not destroy $(C, V \setminus C)$, so $\delta(C)$ is still a min-cut set in $G''$. Since contraction cannot decrease the capacity of any cut, this implies that every min-cut in $G''$ is a min-cut in $G'$, and thus, Algorithm \ref{alg:exp-decomp-min-cut} returns a min-cut.

We now bound the run-time of Algorithm \ref{alg:exp-decomp-min-cut}. 
By Theorem \ref{thm:CX-approximation}, \textsc{CXApproximation} runs in time $O(p)$ to return $k\le 3\lambda$. By Theorem \ref{thm:k-certificate},  \textsc{Certificate} also runs in time $O(p)$ to return a hypergraph 
$G'=(V, E')$ of size $p_1:= \sum_{e \in E'} |e|=O(rkn)=O(r\lambda n)$ and having min-degree $\delta'\ge \lambda$. By Lemma \ref{lem:efficient-expander-decomp}, we obtain an expander decomposition $\cX$ of $G'$ in time 
\[
O\left(p_1^{1+o(1)}\right)
\]
such that the number of hyperedges of $G'$ intersecting multiple parts of $\cX$ is 
$
O(r\phi p_1^{1+o(1)})
$. 
By Claim \ref{clm:efficient-trim-and-shave}, we can implement a $\trim$ operation followed by $3r^2$ $\shave$ operations to obtain a collection $\cX''$ of disjoint subsets of $V$ in time $O(p_1r \cdot 3r^2) = O(p_1r^3)$.   
By Claim \ref{clm:hyperedges-lost-to-trim} and $3r^2$ repeated applications of Claim \ref{clm:hyperedges-lost-to-shave} (similar to the proof of Theorem \ref{thm:number-of-hyperedges-in-mincuts}), we have that 
$
|E''| = O(4r^{9r^2+1}\phi p_1^{1+o(1)}).
$
Since contraction cannot decrease the degree of any vertex, every vertex in $G''=(V'',E'')$ has degree at least $\delta'$. Therefore, the number of vertices $n_2$ in $G''$ satisfies $n_2\delta' \le \sum_{v\in V''}d_{G''}(v)=\sum_{e\in E''}|e|\le r|E''|$, which implies that 
\[
n_2
= O\left( \frac{4r^{9r^2+2}\phi p_1^{1+o(1)}}{\delta'}\right)
= O\left( \frac{4r^{9r^2+2}\phi p_1^{1+o(1)}}{\lambda}\right).
\]
Let $p_2 := \sum_{e \in E''} |e|$. Since contraction cannot increase the size of the hypergraph, we have that $p_2 \leq p_1$. Moreover, contraction does not increase the rank and hence, the rank of $G''$ is at most $r$. Therefore, the run-time of
$\textsc{CXMinCut}$ on $G''$ is  
\begin{align*}
O(p_2 + r\lambda n_2^2) &= O\left(p_1 + \frac{16r^{18r^2+4}\phi^2 p_1^{2+o(1)}}{\lambda}\right) \\
&= O \left(r\lambda n + \frac{r^{19r^2}(\lambda n)^{2+o(1)}\phi^2}{\lambda} \right) \\
&= O\left(r\lambda n + r^{19r^2}\lambda^{1+o(1)}n^{2+o(1)}\phi^2 \right) \\
&= O\left(r\lambda n + r^{20r^2}\lambda^{\frac{r-3}{r-1}+o(1)}n^{2+o(1)}  \right).
\end{align*}
The second equation is by the bound on $p_1$, and the last equation is by the setting of $\phi$ and the fact that $\delta' \geq \lambda$.

The run-time of \textsc{FPZMinCut} on $G''$ is 
\begin{align*}
    \tilde{O}\left(p_2  + n_2^r\right)
    &= \tilde{O}\left(p_1 + \left( \frac{4r^{9r^2+2}\phi p_1^{1+o(1)}}{\lambda}\right)^r  \right) \\
    &= \tilde{O}\left( \lambda n + \frac{r^{11r^3}(\lambda n)^{r(1+o(1)}\phi^r}{\lambda^r} \right) \\
    &= \tilde{O}\left(\lambda n + r^{11r^3}\lambda^{r \cdot o(1)}n^{r(1+ o(1))}\phi^r \right) \\
    &= \tilde{O}\left(\lambda n + \frac{r^{15r^3}n^{r(1+o(1))}}{\lambda^{\frac{r}{r-1}-r \cdot o(1)}} \right).
\end{align*}
Again, the second equation is by the bound on $p_1$, and the last equation is by the setting of $\phi$ and using the fact that $\delta'\ge \lambda$. 

Let $m_2$ be the number of edges in $G''$. We assume that $G''$ is connected (otherwise the problem is trivial). Thus 
we have that $n_2 \leq p_2$, and therefore $(m_2 + n_2) \leq 2p_2$. Since $p_2 \leq rm_2$, this bound is tight, up to a factor of $r$.
Therefore, the run-time of $\textsc{CQMinCut}$ on $G''$ is
\begin{align*}
    \tilde{O}\left(\sqrt{p_2n_2(m_2+n_2)^{1.5}}\right) &= \tilde{O}\left(\sqrt{p_2n_2(2p_2)^{1.5}}\right) \\
    &= \tilde{O}\left(\sqrt{p_1^{2.5}n_2 } \right) \\
    &= \tilde{O}\left(p_1^{1.25}n_2^{0.5} \right) \\
    &= \tilde{O}\left((\lambda n)^{1.25}\left(\frac{4r^{9r^2+2}\phi (\lambda n)^{1+o(1)}}{\lambda}\right)^{0.5} \right) \\
    &=\tilde{O}\left(r^{4.5r^2+1}\lambda^{1.25+o(1)}n^{1.75+o(1)}\left(\frac{6r^2}{\lambda}\right)^{\frac{1}{2(r-1)}} \right) \\
    &=\tilde{O}\left(r^{4.5r^2+2}\lambda^{\frac{5r-7}{4r-4}+o(1)}n^{1.75+o(1)} \right).
\end{align*}
The fourth equation is by the bound on $p_1$, and the second-to-last equation is by the setting of $\phi$ and the fact that $\delta' > \lambda$.

Since the algorithm runs the fastest algorithm among $\textsc{CXMinCut}$, $\textsc{FPZMinCut}$, and $\textsc{CQMinCut}$, the overall runtime is
\begin{align*}
    &\tilde{O}\left(p+p_1^{1+o(1)} + r^3p_1 + \min \left\{
    r\lambda n + r^{20r^2}\lambda^{\frac{r-3}{r-1}+o(1)}n^{2+o(1)},
    \lambda n + \frac{r^{15r^3}n^{r(1+o(1))}}{\lambda^{\frac{r}{r-1}-r \cdot o(1)}},
    r^{4.5r^2+2}\lambda^{\frac{5r-7}{4r-4}+o(1)}n^{\frac{7}{4}+o(1)}  \right\}\right) \\
    & \quad = \tilde{O}\left(p + (r\lambda n)^{1+o(1)} + r^3\lambda n + \min \left\{
    r^{20r^2}\lambda^{\frac{r-3}{r-1}+o(1)}n^{2+o(1)},
    \frac{r^{15r^3}n^{r(1+o(1))}}{\lambda^{\frac{r}{r-1}-r \cdot o(1)}},  r^{4.5r^2+2}\lambda^{\frac{5r-7}{4r-4}+o(1)}n^{\frac{7}{4}+o(1)}  \right\} \right).
\end{align*}
The run-time bound stated in the theorem follows by observing that $\lambda n \le \delta n \le p$, where $\delta$ is the min-degree of the input graph. 
\end{proof}

\begin{remark}\label{remark:det-exp-decomp-min-cut}
We can obtain a deterministic counterpart of Theorem \ref{thm:large-mincut-algo} by
modifying the last step of Algorithm \ref{alg:exp-decomp-min-cut} to only run  \textsc{CXMinCut}. 
The resulting algorithm has a run-time of 
\[
\tilde{O}_r \left(p + (\lambda n)^{1 + o(1)} + \lambda^{\frac{r-3}{r-1}+o(1)}n^{2+o(1)}\right).
\]
\end{remark}

\subsection{Proof of Theorem \ref{theorem:algo}}\label{sec:main-algorithm}
In this section we restate and prove Theorem \ref{thm:min-cut-algorithm}. 

\minCutAlgorithm*
\begin{proof}

We use the algorithm \textsc{CXApproximation} from Theorem \ref{thm:CX-approximation} to obtain a value $k$ such that $\lambda < k \leq 3\lambda$ in $O(p)$ time. If $k \leq 3r(4r^2)^r$, then the result follows by first obtaining a $k$-certificate of the hypergraph with the algorithm $\textsc{Certificate}$ from Theorem \ref{thm:k-certificate} 
and then running Chekuri and Quanrud's algorithm $\textsc{CQMincut}$ from Theorem \ref{thm:CQ-min-cut} on the $k$-certificate. 
Henceforth, we assume that $k > 3r(4r^2)^r$ and consequently, 
$\lambda > r(4r^2)^r$.

Our algorithm finds an approximation for the min-cut capacity, then uses the $k$-certificate from Theorem \ref{thm:k-certificate} to reduce the size of the hypergraph to $O(\lambda n)$. At this point, we either run the min-cut algorithm from Theorem \ref{thm:slow-min-cut} or use the two algorithms we have described in the previous two sections, whichever is faster on the smaller hypergraph. We give a pseudocode of our approach in Algorithm \ref{alg:main-alg}. This algorithm uses the algorithms from Theorems 
\ref{thm:CX-approximation}, \ref{thm:k-certificate}, \ref{thm:new-small-mincut-algo} and \ref{thm:large-mincut-algo}.

\begin{algorithm}
\caption{MinCut$(G)$}\label{alg:main-alg}
\begin{algorithmic}
\State $k \gets \textsc{CXApproximation}(G)$
\State $G' \gets \textsc{Certificate}(G, k)$
\State $(C_1, V\setminus C_1) \gets$ \textsc{SmallSizeMinCut}$(G',r-\log(k/12r)/ \log n)$
\State $(C_2, V \setminus C_2) \gets \textsc{ExpDecompMinCut}(G')$
\State $C \gets \arg\min\{|\delta_G(C_1)|, |\delta_G(C_2)|\}$ 
\State \Return $(C, V \setminus C)$
\end{algorithmic}
\end{algorithm}

We first show that Algorithm \ref{alg:main-alg} returns a min-cut. By Theorem \ref{thm:CX-approximation}, we have that $\lambda <k \le 3\lambda$. Therefore, by Theorem \ref{thm:k-certificate}, the subhypergraph $G'$ is a simple hypergraph and every min-cut in $G$ has capacity $\lambda$ in $G'$ and every other cut has capacity greater than $\lambda$ in $G'$. Thus, the set of min-cuts in $G'$ is exactly the same as the set of min-cuts in $G$. It remains to show that the cut $(C, V \setminus C)$ returned by the algorithm is a min-cut in $G'$. 
We distinguish two cases: (i) Suppose some min-cut in $G'$ has size at most $r-\log (\lambda/4r)/ \log n$. We have that $r-\log (\lambda/4r)/ \log n \leq r- \log (k/12r)/\log n$. Therefore, by Theorem \ref{thm:new-small-mincut-algo}, the call to \textsc{SmallMinCut} returns a min-cut in $G'$, and thus $(C, V \setminus C)$ is a min-cut in $G'$. (ii) Suppose that every min-cut in $G'$ has size greater than $r-\log (\lambda/4r)/ \log n$. Then, by Theorem \ref{thm:large-mincut-algo}, the algorithm $\textsc{ExpDecompMinCut}$ returns a min-cut in $G'$ and thus, $(C, V \setminus C)$ is a min-cut in $G'$. 

We now analyze the runtime of Algorithm \ref{alg:main-alg}. By Theorems \ref{thm:CX-approximation} and \ref{thm:k-certificate}, the computation of $k$ and $G'$ can be done in $O(p)$ time. Also by Theorem \ref{thm:k-certificate}, the algorithm $\textsc{Certificate}$ returns a hypergraph $G'=(V,E')$ with $\vol_{G'}(V) = \sum_{e \in E'} |e| = O(rkn)$. 
By Theorem \ref{thm:new-small-mincut-algo}, the runtime of $\textsc{SmallSizeMinCut}(G', r-\log(k/12r)/\log n)$ is
\[
\tilde{O}_r \left(\left(r-\frac{\log \left(\frac{k}{12r}\right)}{\log n}\right)^{6r}kn \right) = \tilde{O}_r(\lambda n) = \tilde{O}_r(p).
\]
By Theorems \ref{thm:CX-approximation} and \ref{thm:k-certificate}, the min-cut capacity in $G'$ is still $\lambda$. Therefore, by Theorem \ref{thm:large-mincut-algo}, the runtime of   $\textsc{ExpDecompMinCut}(G')$ is 
\begin{align*}
&\tilde{O}_r\left(p + (\lambda n)^{1+o(1)} + \min\left\{
\lambda^{\frac{r-3}{r-1}+o(1)}n^{2+o(1)},
\frac{n^{r+r \cdot o(1)}}{\lambda^{\frac{r}{r-1}}},  \lambda^{\frac{5r-7}{4r-4}+o(1)}n^{\frac{7}{4}+o(1)} \right\}\right) \\ =\, &\hat{O}_r\left(p + \min\left\{
\lambda^{\frac{r-3}{r-1}}n^2, 
\frac{n^r}{\lambda^{\frac{r}{r-1}}}, 
\lambda^{\frac{5r-7}{4r-4}}n^{\frac{7}{4}} 
\right\}\right).\end{align*}
Thus, the overall runtime of the algorithm is 
\[
O(p)+(r\lambda n)^{1+o(1)}+r^{O(r^3)}\min\left\{\lambda^{\frac{r-3}{r-1}+o(1)}n^{2+o(1)} , \left(\frac{n^r}{\lambda^{\frac{r}{r-1}}}\right)\log^{O(1)}n, \lambda^{\frac{5r-7}{4r-4}+o(1)}n^{\frac{7}{4}+o(1)}\log^{O(1)}n\right\} 
\]
\[
= \hat{O}_r \left(p + \min\left\{
\lambda^{\frac{r-3}{r-1}}n^2,
\frac{n^r}{\lambda^{\frac{r}{r-1}}}, 
\lambda^{\frac{5r-7}{4r-4}}n^{\frac{7}{4}} \right\}\right).\]
\end{proof}

\begin{remark}
We obtain the deterministic counterpart of Theorem \ref{theorem:algo} by using the deterministic version of \textsc{ExpDecompMinCut} in Algorithm \ref{alg:main-alg} (that was discussed in Remark \ref{remark:det-exp-decomp-min-cut}) and by using the deterministic algorithm for \textsc{SmallMinCut} described in Appendix \ref{sec:det-small-size-min-cut} instead of the algorithm from section \ref{sec:small-size-min-cut}. The resulting algorithm has a run-time of 
\[
\hat{O}_r\left(p + \min\left\{\lambda n^2, \lambda^{\frac{r-3}{r-1}}n^2 + \frac{n^{r}}{\lambda} \right\}\right).
\]
\end{remark}

\section*{Acknowledgement}

This project has received funding from the European Research Council (ERC) under the European Union's Horizon 2020 research and innovation programme under grant agreement No 715672. The last two authors are also supported by the Swedish Research Council (Reg. No. 2015-04659 and 2019-05622). Karthekeyan and Calvin are supported in part by NSF grants CCF-1814613 and CCF-1907937.

\bibliography{biblio}

\newcommand{\etalchar}[1]{$^{#1}$}
\begin{thebibliography}{BBG{\etalchar{+}}20}

\bibitem[BBG{\etalchar{+}}20]{BernsteinBNPSS20}
Aaron Bernstein, Jan van~den Brand, Maximilian~Probst Gutenberg, Danupon
  Nanongkai, Thatchaphol Saranurak, Aaron Sidford, and He~Sun.
\newblock Fully-dynamic graph sparsifiers against an adaptive adversary.
\newblock {\em CoRR}, abs/2004.08432, 2020.

\bibitem[BGS20]{BernsteinGS20scc}
Aaron Bernstein, Maximilian~Probst Gutenberg, and Thatchaphol Saranurak.
\newblock Deterministic decremental reachability, scc, and shortest paths via
  directed expanders and congestion balancing.
\newblock In {\em {FOCS}}. {IEEE} Computer Society, 2020.

\bibitem[CGL{\etalchar{+}}20]{cglnps20}
Julia Chuzhoy, Yu~Gao, Jason Li, Danupon Nanongkai, Richard Peng, and
  Thatchaphol Saranurak.
\newblock A deterministic algorithm for balanced cut with applications to
  dynamic connectivity, flows, and beyond.
\newblock In {\em {FOCS}}. {IEEE} Computer Society, 2020.

\bibitem[CQ21]{CQ21}
Chandra Chekuri and Kent Quanrud.
\newblock {Isolating Cuts, (Bi-)Submodularity, and Faster Algorithms for
  Connectivity}.
\newblock In {\em {ICALP}}, pages 50:1--50:20, 2021.

\bibitem[CQX19]{CQX19}
Chandra Chekuri, Kent Quanrud, and Chao Xu.
\newblock {LP Relaxation and Tree Packing for Minimum $k$-cuts}.
\newblock In {\em {SOSA}}, pages 7:1--7:18, 2019.

\bibitem[CX17]{CX17}
Chandra Chekuri and Chao Xu.
\newblock Computing minimum cuts in hypergraphs.
\newblock In {\em {SODA}}, pages 1085--1100. {SIAM}, 2017.

\bibitem[CX18]{ChekuriX18}
Chandra Chekuri and Chao Xu.
\newblock Minimum cuts and sparsification in hypergraphs.
\newblock {\em SIAM Journal on Computing}, 47(6):2118--2156, 2018.

\bibitem[CXY19]{CXY19}
Karthekeyan Chandrasekaran, Chao Xu, and Xilin Yu.
\newblock Hypergraph $k$-cut in randomized polynomial time.
\newblock {\em Mathematical Programming (Preliminary version in SODA 2018)},
  Nov 2019.

\bibitem[Din06]{Dinitz06}
Yefim Dinitz.
\newblock Dinitz’algorithm: The original version and even’s version.
\newblock In {\em Theoretical computer science}, pages 218--240. Springer,
  2006.

\bibitem[FNY{\etalchar{+}}20]{FNYSY20}
Sebastian Forster, Danupon Nanongkai, Liu Yang, Thatchaphol Saranurak, and
  Sorrachai Yingchareonthawornchai.
\newblock Computing and testing small connectivity in near-linear time and
  queries via fast local cut algorithms.
\newblock In {\em SODA}, pages 2046--2065. {ACM/SIAM}, 2020.

\bibitem[FPZ19]{FPZ19}
Kyle Fox, Debmalya Panigrahi, and Fred Zhang.
\newblock Minimum cut and minimum k-cut in hypergraphs via branching
  contractions.
\newblock In {\em {SODA}}, pages 881--896. {SIAM}, 2019.

\bibitem[Fuk13]{F10}
Takuro Fukunaga.
\newblock Computing minimum multiway cuts in hypergraphs.
\newblock {\em Discrete Optimization}, 10(4):371--382, 2013.

\bibitem[GKP17]{GKP17}
Mohsen Ghaffari, David Karger, and Debmalya Panigrahi.
\newblock Random contractions and sampling for hypergraph and hedge
  connectivity.
\newblock In {\em {SODA}}, page 1101–1114. {ACM/SIAM}, 2017.

\bibitem[GMT15]{GMT15}
Sudipto Guha, Andrew McGregor, and David Tench.
\newblock Vertex and hyperedge connectivity in dynamic graph streams.
\newblock In {\em SIGMOD/PODS}, pages 241--247. {ACM}, 2015.

\bibitem[GMW20]{GMW19}
Pawel Gawrychowski, Shay Mozes, and Oren Weimann.
\newblock {Minimum Cut in $O(m \log ^2 n)$ Time}.
\newblock In {\em {ICALP}}, pages 57:1--57:15, 2020.

\bibitem[GNT20]{GhaffariNT20}
Mohsen Ghaffari, Krzysztof Nowicki, and Mikkel Thorup.
\newblock Faster algorithms for edge connectivity via random 2-out
  contractions.
\newblock In {\em {SODA}}. {ACM/SIAM}, 2020.

\bibitem[GRST21]{GoranciRST20hierarchy}
Gramoz Goranci, Harald R{\"{a}}cke, Thatchaphol Saranurak, and Zihan Tan.
\newblock The expander hierarchy and its applications to dynamic graph
  algorithms.
\newblock {\em {SODA}}, pages 2212--2228, 2021.

\bibitem[HRW17]{HenzingerRW17}
Monika Henzinger, Satish Rao, and Di~Wang.
\newblock Local flow partitioning for faster edge connectivity.
\newblock In {\em {SODA}}, pages 1919--1938. {ACM/SIAM}, 2017.

\bibitem[Kar93]{Kar93}
David Karger.
\newblock {Global Min-Cuts in RNC, and Other Ramifications of a Simple Min-Cut
  Algorithm}.
\newblock In {\em {SODA}}, page 21–30. {ACM/SIAM}, 1993.

\bibitem[Kar00]{Kar00}
David~R. Karger.
\newblock Minimum cuts in near-linear time.
\newblock {\em J. {ACM}}, 47(1):46--76, 2000.
\newblock announced at STOC'96.

\bibitem[KK15]{KK15}
Dmitry Kogan and Robert Krauthgamer.
\newblock Sketching cuts in graphs and hypergraphs.
\newblock In {\em {ITCS}}, pages 367--376, 2015.

\bibitem[KS96]{KS96}
David Karger and Clifford Stein.
\newblock A new approach to the minimum cut problem.
\newblock {\em Journal of the ACM}, 43(4):601--640, July 1996.

\bibitem[KT15]{KT15}
Ken{-}ichi Kawarabayashi and Mikkel Thorup.
\newblock Deterministic edge connectivity in near-linear time.
\newblock In {\em {STOC}}, pages 665--674. {ACM}, 2015.

\bibitem[KT19]{KawarabayashiT19}
Ken{-}ichi Kawarabayashi and Mikkel Thorup.
\newblock Deterministic edge connectivity in near-linear time.
\newblock {\em J. {ACM}}, 66(1):4:1--4:50, 2019.

\bibitem[KW96]{KW96}
Regina Klimmek and Frank Wagner.
\newblock A simple hypergraph min cut algorithm.
\newblock Technical Report B 96-02, Institute Of Computer Science, Freie
  Universitat, 1996.

\bibitem[LNP{\etalchar{+}}21]{LNPSY21}
Jason Li, Danupon Nanongkai, Debmalya Panigrahi, Thatchaphol Saranurak, and
  Sorrachai Yingchareonthawornchai.
\newblock Vertex connectivity in poly-logarithmic max-flows.
\newblock unpublished, 2021.

\bibitem[LP20]{LP20}
Jason Li and Debmalya Panigrahi.
\newblock {Deterministic Min-cut in Poly-logarithmic Max-flows}.
\newblock In {\em {FOCS}}. {IEEE} Computer Society, 2020.

\bibitem[Mat93]{Matula93}
David~W. Matula.
\newblock A linear time 2+epsilon approximation algorithm for edge
  connectivity.
\newblock In {\em {SODA}}, pages 500--504. {ACM/SIAM}, 1993.

\bibitem[MN20]{MukhopadhyayN20}
Sagnik Mukhopadhyay and Danupon Nanongkai.
\newblock Weighted min-cut: sequential, cut-query, and streaming algorithms.
\newblock In {\em {STOC}}, pages 496--509. {ACM}, 2020.

\bibitem[MW00]{MW00}
Wai-Kei Mak and Martin D.~F. Wong.
\newblock {A fast hypergraph min-cut algorithm for circuit partitioning}.
\newblock {\em {Integration: the VLSI Journal}}, 30(1):1--11, 2000.

\bibitem[NI92]{NagamochiI92}
Hiroshi Nagamochi and Toshihide Ibaraki.
\newblock Computing edge-connectivity in multigraphs and capacitated graphs.
\newblock {\em SIAM Journal on Discrete Mathematics}, 5(1):54--66, 1992.

\bibitem[NS17]{ns17}
Danupon Nanongkai and Thatchaphol Saranurak.
\newblock Dynamic spanning forest with worst-case update time: adaptive, las
  vegas, and ${O}(n^{1/2 - \epsilon})$-time.
\newblock In {\em {STOC}}, pages 1122--1129. {ACM}, 2017.

\bibitem[NSY19]{NanongkaiSY19}
Danupon Nanongkai, Thatchaphol Saranurak, and Sorrachai Yingchareonthawornchai.
\newblock Breaking quadratic time for small vertex connectivity and an
  approximation scheme.
\newblock In {\em {STOC}}, pages 241--252. {ACM}, 2019.

\bibitem[Que98]{queyranne98}
Maurice Queyranne.
\newblock Minimizing symmetric submodular functions.
\newblock {\em Mathematical Programming}, 82(1-2):3--12, 1998.

\bibitem[RSW18]{RubinsteinSW18}
Aviad Rubinstein, Tselil Schramm, and S.~Matthew Weinberg.
\newblock Computing exact minimum cuts without knowing the graph.
\newblock In {\em {ITCS}}, pages 39:1--39:16, 2018.

\bibitem[Sar21]{Sar20}
Thatchaphol Saranurak.
\newblock {A Simple Deterministic Algorithm for Edge Connectivity}.
\newblock In {\em SOSA}. {SIAM}, 2021.

\bibitem[SW19]{sw19}
Thatchaphol Saranurak and Di~Wang.
\newblock Expander decomposition and pruning: Faster, stronger, and simpler.
\newblock In {\em {SODA}}, pages 2616--2635. {SIAM}, 2019.

\bibitem[Tho08]{Th08}
Mikkel Thorup.
\newblock {Minimum $k$-way Cuts via Deterministic Greedy Tree Packing}.
\newblock In {\em {STOC}}, pages 159--166. {ACM}, 2008.

\bibitem[Wul17]{w17}
Christian Wulff{-}Nilsen.
\newblock Fully-dynamic minimum spanning forest with improved worst-case update
  time.
\newblock In {\em {STOC}}, pages 1130--1143. {ACM}, 2017.

\end{thebibliography}

\appendix

\section{Deterministic Algorithm for Finding Small-sized Min-Cut}\label{sec:det-small-size-min-cut}
In this section, we consider the problem of finding a min-cut subject to an upper bound on the size of the smaller side of the cut. 
In particular, we prove the following algorithmic result, which can be viewed as a deterministic version of Theorem \ref{thm:new-small-mincut-algo}. 

\begin{lemma}\label{lem:small-mincut-algo}
Let $s \leq r$ be a positive integer, and let $G=(V,E)$ be an $r$-rank $n$-vertex hypergraph with $m$ hyperedges that has a min-cut $(C, V \setminus C)$ with $|C| \leq s$. Then, there exists a deterministic algorithm which takes $G$ and $s$ as input and runs in time $O(2^{2s}n^s + 2^rm)$ to return a min-cut of $G$.
\end{lemma}

\begin{proof}
We begin by defining some functions that will be useful in our algorithm and analysis. 
For $S'\subseteq S \subseteq V$, let 
\begin{align*}
    g(S) &:= |\{e \in E \colon S \subseteq e\}|,\\
    g'(S) &:= \begin{cases} 1 &\text{if } S \in E \\
    0 &\text{otherwise}
    \end{cases}\\ 
    g_S(S') &:= |\{e \in E \colon e \cap S = S'\}|, \text{ and}\\
    h_S(S') &:= |\{e \in E \colon e \cap S = S' \text{ and } e \setminus S \neq \emptyset\}|.
\end{align*}

We first show that we can compute $g(S)$ and $g'(S)$ for all subsets $S\subseteq V$ of size at most $s$ in time $O(n^s + 2^sm)$ using Algorithm \ref{alg:min-cut-helper}.

\begin{algorithm}
\caption{MinCutHelper$(G,s)$}\label{alg:min-cut-helper}
\begin{algorithmic}
\ForAll{$S \subseteq V \colon |S| \leq s$}
\State $g[S] \gets 0$
\State $g'[S] \gets 0$
\EndFor
\ForAll{$e \in E$}
\ForAll{$S \subseteq e \colon |S| \leq s$}
\State $g[S] \gets g[S] + 1$
\EndFor
\If{$|e| \leq s$}
\State $g'[e] \gets 1$
\EndIf
\EndFor
\State \Return $(g,g')$
\end{algorithmic}
\end{algorithm}

We note that Algorithm \ref{alg:min-cut-helper} correctly computes $g(S)$ for each $S$ since $g[S]$ will be incremented exactly once for each $e$ containing $S$. It also correctly computes $g'(S)$, since $g'[S]$ will be set to $1$ for an $S$ of size at most $s$ if and only if $e \in E$. Furthermore the algorithm spends $O(n^s)$ time initializing the arrays $g$ and $g'$ and $O(2^rm)$ iterating over subsets of the hyperedges (since a hyperedge of size at most $r$ has at most $2^r$ subsets). Therefore, the algorithm runs in time $O(n^s + 2^rm)$.

Next, we solve the min $s$-sized cut problem using Algorithm \ref{alg:small-min-cut}. 

\begin{algorithm}
\caption{SmallMinCut$(G,s)$}\label{alg:small-min-cut}
\begin{algorithmic}
\State $(g,g') \gets \textsc{MinCutHelper(G,s)}$
\ForAll{$S \subseteq V \colon |S| \leq s$}
\ForAll{$S' \subset S$}
\State $g_S[S'] = \sum_{j=0}^{|S|-|S'|} (-1)^j \sum_{S'' \colon S' \subseteq S'' \subseteq S \text{ and } |S''| = |S'|+j} g[S'']$
\State $h_S[S'] = g_S[S'] - g'[S']$
\EndFor
\State $r[S] \gets \sum_{S' \subseteq S} h_S[S']$
\EndFor
\State $C \gets \text{argmin}_{S \subseteq V \colon |S| \leq s} r[S]$
\State \Return $(C, V \setminus C)$
\end{algorithmic}
\end{algorithm}

To see that Algorithm \ref{alg:small-min-cut} is correct, we note that, by the principle of inclusion-exclusion, it correctly computes $g_S(S')$. 
To argue this, we show that a hyperedge $e$ with $e \cap S = T \supset S'$ will not be counted in our summation for $g_S(S')$. Let $t = |T| - |S'|$. A set $S''$ such that $S' \subseteq S'' \subseteq S$ counts $e$ in $g_S(S'')$ if and only if $S'' \subseteq T$. For any $j \in \{0, \dots, t\}$, the number of sets $S''$ of size $|S'|+j$ such that $S' \subseteq S'' \subseteq T$ is $\binom{t}{j}$. Thus, the total contribution of $e$ to our summation for $g_S[S']$ is
\[
\sum_{j=0}^{t}(-1)^j \binom{t}{j}= 0.
\]
A similar argument shows that the hyperedges which need to be counted in $g_S(S')$ are counted exactly once by the expression. Since $h_S(S')$ is $g_S(S')$ if $S' \not\in e$ and $g_S(S') - 1$ otherwise, the algorithm correctly computes $h_S(S')$ as well. We note that any hyperedge $e \in \delta(S)$ must be counted in $h_S(S')$ for exactly one set $S' \subseteq S$ (namely $h_S(e \cap S)$). Therefore, $|\delta(S)| = \sum_{S' \subseteq S} h_S(S')$. Thus, we have that $r[S]$ correctly stores $\delta(S)$ for each $S \subseteq V$ of size at most $s$. Therefore, since the algorithm returns a set $S$ with minimum $r[S]$, it returns a min $s$-sized cut.

The algorithm's outer \textbf{for} loop iterates over all subsets of $V$ of size at most $s$. The number of such subsets is $n^s$. The inner for loop iterates over the subsets of a subset $S$ of size at most $s$. The number of such subsets is at most $2^s$. Thus, the inner \textbf{for} loop is executed $O(2^sn^s)$ times. Each iteration of the inner loop can be implemented to run in $O(2^s)$ time, since this is an upper bound on the number of sets $S''$ with $S' \subseteq S'' \subseteq S$, and thus on the number of terms in the double summation. Therefore, the total runtime of the outer loop is $O(2^{2s}n^s)$. Adding this to the runtime we computed for Algorithm \ref{alg:min-cut-helper} gives us an overall runtime of $O(2^{2s}n^s + 2^rm)$.
\end{proof} 

\section{Min-cut non-triviality for hypergraph} \label{appn:hypergraph-nontrivial}

In this section, we construct an example where $\lambda > |V|$ and even though every min-cut is non-trivial, the total number of hyperedges that take part in the union of all min-cuts is only a constant fraction of the total number of hyperedges.

If we consider a complete graph on $n$ vertices, we see that every edge takes part in a min-cut. Hence the total number of edges that take part in the union of min-cuts is $m = \Theta(n^2)$. The issue with this example is that all min-cuts here are \textit{trivial}, i.e., every min-cut has a single vertex as one side. 
We extend this example to hypergraphs such that (i) the capacity of the min-cut is $\Omega(n)$, (ii) all min-cuts are non-trivial, i.e., each min-cut has at least two vertices on each side, and (iii) the number of hyperedges taking part in the union of min-cuts is approximately $m$. We prove the following lemma.

\begin{lemma}\label{lem:hypergraph-nontrivial}
Let $n \geq 100$ be an even integer. Then, there is a simple hypergraph $G= (V, E)$ with $n + 3$ vertices and $m=\Theta(n^2)$ hyperedges such that 
\begin{enumerate}
    \item $\lambda(G)\ge n+4$, 
    \item every min-cut has at least two vertices on both sides, and 
    \item the number of hyperedges in the union of all min-cuts is $\Theta(m)$.
\end{enumerate}
\end{lemma}
\begin{proof}
We construct such a hypergraph $G$ now. 
The vertex set $V$ consists of $n+3$ vertices: $n/2$ vertices $\{u_1, \ldots, u_n$\}, $n/2$ vertices $\{v_1, \ldots, v_n$\}, and three special vertices $\{a,b,c\}$. We 
add the following 5-uniform hyperedges: For every pair of integers $1 \leq i < j \leq n/2$
, we add the three hyperedges $\{u_i, v_i, u_j, v_j, a\}, \{u_i, v_i, u_j, v_j, b\}$ and $\{u_i, v_i, u_j, v_j, c\}$. We also add edges $\{u_i, v_i\}$ for each $i \in [n/2]$. The next three claims complete the proof of the lemma. 

\end{proof}

\begin{claim} \label{clm:appn-1}
The number of hyperedges in $G$ is $\Theta(n^2)$.
\end{claim}

\begin{proof}
For every two distinct pairs $\{u_i, v_i\}$ and $\{u_j, v_j\}$, there are three 5-uniform hyperedges. Hence the total number of 5-uniform hyperedges is $\Theta(n^2)$. The number of edges is $O(n)$. Hence the total number of hyperedges is $m = \Theta(n^2)$.
\end{proof}

\begin{claim}\label{clm:appn-2}
The min-cut is of size $\lambda > n +3$. Moreover, each min-cut contains a pair $\{u_i, v_i\}$ on one side, and the rest of the vertices on the other side.
\end{claim}

\begin{proof}

Let $(C, V \setminus C)$ be a min-cut in $G$. We first note that for every $i \in [n/2]$, the cut set $\delta(C)$ cannot contain the edge $\{u_i, v_i\}$. To see this, suppose $u_i \in C$, $v_i \in V \setminus C$, and assume without loss of generality that $|C| \leq |V \setminus C|$. We note that $V \setminus C \setminus \{v_i\}$ is non-empty since $|C| \leq |V \setminus C|$, and therefore $(C \cup \{v_i\}, V \setminus C \setminus \{v_i\})$ is a cut. Furthermore, this cut has strictly smaller capacity than $(C, V \setminus C)$, since it cuts no new hyperedges and it does not cut the edge $\{u_i, v_i\}$. Thus, we conclude that $(C, V \setminus C)$ does not cut any edges.

We note that one side of the min-cut $(C, V \setminus C)$ must include at most $1$ vertex from among $\{a,b,c\}$. Without loss of generality assume that $|C \cap \{a,b,c\}| \leq 1$. We consider two cases. We will show that if $|C \cap \{a,b,c\} = 0$, then $C$ must contain only a single pair $\{u_i, v_i\}$ and $|\delta(C)| > n+3$. Then we will show that the case where $|C \cap \{a,b,c\}| = 1$ cannot actually occur, because $C$ is a min-cut.

Case 1: Suppose $|C \cap \{a,b,c\}| = 0$. Let $k$ be the number of pairs $\{u_i, v_i\}$ with $u_i, v_i \in C$. Any hyperedge intersecting one of these pairs is cut by $(C, V \setminus C)$. The number of such hyperedges is minimized when $k=1$. We note that for any pair $\{u_i, v_{i}\}$, there are $n/2-1$ pairs $\{u_j, v_j\}$ that share a $5$-uniform hyperedge with $\{u_i, v_i\}$, and each pair shares three such hyperedges. Hence, $|\delta(C)| = |\delta(\{v_i, v_{i+1}\})| = 3(n/2-1) = 3n/2-3 > n+3$. 

Case 2: Suppose $|C \cap \{a,b,c\}| = 1$. Without loss of generality, assume $a \in C$. Let $k$ be the number of pairs $\{u_i, v_i\}$ with $u_i, v_i \in C$. We note that any hyperedge containing a pair $\{u_i, v_i\} \subseteq C$ and containing either vertex $b$ or vertex $c$ is in $\delta(C)$. There are $k(k-1)/2$ ways to pick two pairs inside $C$ and $k(n/2-k)$ ways to pick a pair inside $C$ and a pair outside $C$. Hence, $\delta(C)$ has $2(k(n/2-k)+k(k-1)/2)$ hyperedges which contain $b$ or $c$. We also note that any hyperedge containing a pair $\{u_i, v_i\}$ with $u_i, v_i \in V \setminus C$ as well as the vertex $a$ is in $\delta(C)$. Thus, $\delta(C)$ has $k(n/2-k) + (n/2-k)(n/2-k-1)/2$ hyperedges which contain $a$. Therefore, we have that 
$
|\delta(C)| = 2(k(n/2-k)+k(k-1)/2) + k(n/2-k) + (n/2-k)(n/2-k-1)/2 \\
$
We note that for every $0 \leq k \leq n/2$, $\max\{k(k-1)/2, (n/2-k)(n/2-k-1)/2\} \geq (n/4)(n/4-1)/2 = (n^2-4n)/32 > 3n/2-3$. Thus, we have that $|\delta(C)| > 3n/2-3$. Since we showed in the previous paragraph that there is a cut in $G$ of capacity $3n/2-3$, this contradicts the fact that $(C, V \setminus C)$ is a min-cut, and we conclude that this case cannot occur.
\end{proof}

\begin{claim}\label{clm:appn-3}
The number of edges in the union of all min-cuts is $\Theta(n^2)$.
\end{claim}

\begin{proof}
This follows from the proof of Claims \ref{clm:appn-1} and \ref{clm:appn-2}. The hyperedges which take part in min-cuts are the 5-uniform hyperedges, and every such hyperedge takes part in some min-cut. Hence the claim follows.
\end{proof}

\section{Tight example for the structural theorem} \label{appn:tight-condn-2}

In this section, we show a tight example for Conclusion \ref{itm:few-edges} in  Theorem \ref{thm:number-of-hyperedges-in-mincuts}. 
Our example shows that there is no room to improve the factor $\lambda^{-1/(r-1)}$ in Conclusion \ref{itm:few-edges} of \Cref{thm:number-of-hyperedges-in-mincuts}. 
We prove the following lemma.

\begin{lemma}\label{lem:tight-condn-2}
There is a $r$-uniform simple hypergraph $G = (V, E)$ on $n$ vertices and $m$ hyperedges that satisfies the following conditions: \begin{enumerate}
    \item $\lambda(G) > n > r(4r^2)^r$, 
    \item the smaller side of every min-cut has size $\sqrt n$, and 
    \item the number of hyperedges in the union of all min-cuts is $\Theta(m/\lambda^{1/(r-1)})$.
\end{enumerate}
\end{lemma}

\begin{proof}
Consider the complete $r$-uniform hypergraph $G_1=(V_1,E_1)$ on $\sqrt{n}$ vertices. The degree of each vertex is $ \Theta(n^{\frac{r-1}{2}})$. We make the following claim about $G_1$. 

\begin{claim} \label{clm:appn-4}
Every min-cut in $G_1$ has the smaller side to be a single vertex. Moreover, the capacity of the min-cut $\lambda(G_1) = \Theta(n^{\frac{r-1}{2}})$.
\end{claim}

\begin{proof}
Consider a cut which has (say) $t$ vertices on one side and $\sqrt n- t$ vertices on the other side. The capacity of this cut is
\[
\sum_{a >0, b >0: a+ b = r}\binom{t}{a}\binom{\sqrt n - t}{b}.
\]
The above quantity is minimized when $t = 1$. Hence, the capacity of a min-cut is $\binom{\sqrt n -1}{r -1} = \Theta(n^{\frac{r-1}{2}})$.
\end{proof}

Now, we replace each vertex $u\in V_1$ with a set $Q_u$ of $\sqrt{n}$ vertices. We do this in a way such that every vertex $v \in Q_u$ has degree $\Theta(n^{\frac{r-2}{2}})$ induced by the edges in $E_1$. This can be done as the total number of hyperedges incident on $u$ is $\lambda(G_1)$. Because $Q_u$ has $\sqrt n$ many vertices, we can replace $u$ by a vertex $v \in Q_u$ in $\lambda(G_1)/\sqrt n = \Theta(n^{\frac{r-2}{2}})$ hyperedges. This ensures that each vertex $v \in Q_u$ has degree $\Theta(n^{\frac{r-2}{2}})$ induced by the hyperedges in $E_1$.

Having done so, we add all $r$-uniform hyperedges inside $Q_u$ to make each $Q_u$ a complete $r$-uniform hypergraph. We call this set of hyperedges as $E_{Q_u}$. This completes the description of $G=(V, E)$ on vertex set $V = \bigcup_{u \in V_1}Q_u$ of size $n$ and $E = E_1 \cup \left(\bigcup_{u \in V_1}E_{Q_u}\right)$. Now we make the following claim.

\medskip
\begin{claim} \label{clm:appn-5}
Every min-cut in $G$ is of the form $(Q_u, V\setminus Q_u)$ for some $u\in V_1$.
\end{claim}

\begin{proof}
By a similar argument as in Claim \ref{clm:appn-4}, we see that the min-cut in the hypergraph $G_u=(Q_u, E_{Q_u})$ is obtained when a single vertex is cut from the set $Q_u$, and the capacity of the min-cut $\lambda(G_u) = \lambda(G_1)= \Theta(n^{\frac{r-1}{2}})$. By our construction of $G$, any cut of $G$ that cuts $G_u$ has capacity at least $\lambda(G_u) + \Theta(n^{\frac{r-2}{2}}) > \lambda(G_1)$. However, as proven in Claim \ref{clm:appn-4}, the capacity of a cut of the form $(Q_u, V \setminus Q_u)$ is $\lambda(G_1)$.
\end{proof}

\medskip
Claim \ref{clm:appn-5} implies that $\lambda:=\lambda(G) = \lambda(G_1) = \Theta(n^{\frac{r-1}{2}})$.
The number of hyperedges in $G_1$ is $m_1=\binom{\sqrt{n}}{r}$. The number of hyperedges in $G_u$ is $m_2'=\binom{\sqrt{n}}{r}$. Thus, the number of hyperedges in $G$ is 
\[
m=m_1 + m_2' \sqrt{n}=\binom{\sqrt{n}}{r}(1+\sqrt{n})=\Theta\left(n^{\frac{r+1}{2}}\right).
\]
Claim \ref{clm:appn-5} also implies that the number of hyperedges in the union of all min-cuts is exactly $m_1$. Now, we observe that 
\[
m_1 = \Theta\left(n^{\frac{r}{2}}\right) = \Theta\left(\frac{m}{\lambda^{\frac{1}{r-1}}}\right).
\]
Moreover, in this graph $G$, we have that $\lambda > n > r(4r^2)^r$ by picking $n$ to be large enough and $r\ge 3$. 

\end{proof}

\end{document}